\documentclass[11pt, draftcls,onecolumn]{IEEEtran}

\def\withcolors{0}
\def\withnotes{0}

\usepackage{ccanonne}

\newcommand{\citet}[1]{\cite{#1}}
\newcommand{\citep}[1]{\cite{#1}}

\newcommand{\tst}{t}  %
\newcommand{\ts}{r} %

\newcommand{\signA}{{\bf{}1}^\star}
\newcommand{\signB}{{\bf{}0}^\star}

\newcommand{\vecu}{u}

\newcommand{\italicparagraph}[1]{\medskip%
\noindent\textit{#1}}

\title{Interactive Inference under Information Constraints}

\ifCLASSINFOpdf
\else
\fi

\begin{document}

\author{
  \IEEEauthorblockN{Jayadev Acharya\IEEEauthorrefmark{1},}
  \IEEEauthorblockN{Cl\'{e}ment L. Canonne,\IEEEauthorrefmark{2}}
  \IEEEauthorblockN{Yuhan Liu,\IEEEauthorrefmark{1}}
  \IEEEauthorblockN{Ziteng Sun,\IEEEauthorrefmark{1}}
  and \IEEEauthorblockN{Himanshu Tyagi\IEEEauthorrefmark{3}}
}

\maketitle 
{
\renewcommand{\thefootnote}{}
  \footnotetext{\IEEEauthorblockA{\IEEEauthorrefmark{1}Cornell University. Emails: \{acharya, yl2976, zs335\}@cornell.edu}\\
      \indent\IEEEauthorblockA{\IEEEauthorrefmark{2}University of Sydney. Email: clement.canonne@sydney.edu.au}\\
      \indent\IEEEauthorblockA{\IEEEauthorrefmark{3}The Department of Electrical Communication Engineering,
        Indian Institute of Science, Bangalore 560012, India. Email: htyagi@iisc.ac.in}

 }
\renewcommand{\thefootnote}{\arabic{footnote}}
\setcounter{footnote}{0}

\maketitle

\begin{abstract}
  We study the role of interactivity in distributed statistical inference under information constraints, e.g., communication constraints and local differential privacy. We focus on the tasks of goodness-of-fit testing and estimation of discrete distributions. From prior work, these tasks are well understood under noninteractive protocols. Extending these approaches directly for interactive protocols is difficult due to correlations that can build due to interactivity; in fact, gaps can be found in prior claims of tight bounds of distribution estimation using interactive protocols.
We propose a new approach to handle this correlation and establish a unified method to establish lower bounds for both tasks. As an application, we obtain optimal bounds for both estimation and testing under local differential privacy and communication constraints. We also provide an example of a natural testing problem where interactivity helps.

%
%
%
%
%
%

%
%
%
%
%
%

\end{abstract}

\section{Introduction}
  \label{sec:intro}
  Classical statistics focuses on algorithms that are data-efficient. 
Recent years have seen revived interest in a different set of constraints for distributed
statistics: local constraints on the amount of \emph{information} that can be extracted from each \newest{data point}. These local constraints can be communication 
constraints, where each \newest{data point} must be expressed
using a fixed number of bits; privacy constraints, where each user
holding a sample seeks to reveal \newer{as little as possible} about \newest{it}; 
and many others, such as noisy communication
channels, limited types of measurements, or quantization schemes. Our focus in this work
is on statistical inference under such local constraints, when
interactive protocols are allowed.

We study the strengths and limitations of
interactivity for statistical inference under local information
constraints for two fundamental inference tasks for discrete distributions:
\emph{learning} (density estimation) and \emph{identity testing} (goodness-of-fit) \newest{under total variation distance}.
For these tasks, prior work 
gives a good understanding of the number of samples needed in
noninteractive setting, including a precise dependence on the 
information constraints under consideration. However, the following question
remains largely open:

\begin{quote}\it
  Does interactivity help for learning and testing in total variation distance when the data is
  subject to local information constraints, and, if so, for which type
  of constraints? 
\end{quote}

In this work, we resolve this question by establishing lower bounds
that hold for general channel families (modeling local information constraints). 
We show that interaction does not help for learning and testing under
communication constraints or local privacy constraints. Several prior works
have claimed a subset of these results, but we exhibit technical gaps in most of them (with the important exceptions of~\citet{AJM:20} and~\citet{BB:20}, which both obtain a tight bound
for testing under local privacy constraints). 
These gaps stem from
the difficulty in handling the correlation that builds due to interaction.
Our lower bound explicitly handles this correlation and is based on examining
how effectively one can 
exploit this correlation in spite of the local constraints.
 Furthermore, our lower bounds allow us to identify a family of channels
 for which interaction strictly helps in identity testing,
establishing the first separation between
interactive and noninteractive protocols for distributed
goodness-of-fit.

  \subsection{The setting}\label{sec:setting}
  We now describe the general framework of distributed inference under local
information constraints and then specialize it to two canonical tasks:
\newest{estimation and  testing}.

\begin{figure}[ht!]\centering
 \scalebox{.7}{\begin{tikzpicture}[->,>=stealth',shorten >=1pt,auto,node distance=20mm, semithick]
  \node[circle,draw,minimum size=13mm] (A) {$X_1$};
  \node[circle,draw,minimum size=13mm] (B) [right of=A] {$X_2$};
  \node[circle,draw,minimum size=13mm] (BB) [right of=B] {$X_3$};
  \node (C) [right of=BB] {$\dots$};
  \node[circle,draw,minimum size=13mm] (DD) [right of=C] {$X_{\ns-2}$};
  \node[circle,draw,minimum size=13mm] (D) [right of=DD] {$X_{\ns-1}$};
  \node[circle,draw,minimum size=13mm] (E) [right of=D] {$X_\ns$};
  
  \node[rectangle,draw,minimum width=13mm,minimum height=7mm,fill=gray!20!white] (WA) [below of=A] {$W_1$};
  \node[rectangle,draw,minimum width=13mm,minimum height=7mm,fill=gray!20!white] (WB) [below of=B] {$W_2$};
  \node[rectangle,draw,minimum width=13mm,minimum height=7mm,fill=gray!20!white] (WBB) [below of=BB] {$W_3$};
  \node[rectangle,minimum width=13mm,minimum height=7mm,fill=none] (WC) [below of=C] {$\dots$};
  \node[rectangle,draw,minimum width=13mm,minimum height=7mm,fill=gray!20!white] (WDD) [below of=DD] {$W_{\ns-1}$};
  \node[rectangle,draw,minimum width=13mm,minimum height=7mm,fill=gray!20!white] (WD) [below of=D] {$W_{\ns-1}$};
  \node[rectangle,draw,minimum width=13mm,minimum height=7mm,fill=gray!20!white] (WE) [below of=E] {$W_\ns$};
  
  \node[draw,dashed,fit=(WA) (WB) (WBB) (WC) (WDD) (WD) (WE)] {};
  \node[circle,draw,minimum size=13mm,fill=white] (YA) [below of=WA] {$Y_1$};
  \node[circle,draw,minimum size=13mm,fill=white] (YB) [below of=WB] {$Y_2$};
  \node[circle,draw,minimum size=13mm,fill=white] (YBB) [below of=WBB] {$Y_3$};
  \node[circle,minimum size=13mm] (YC) [below of=WC] {$\dots$};
  \node[circle,draw,minimum size=13mm,fill=white] (YDD) [below of=WDD] {$Y_{\ns-2}$};
  \node[circle,draw,minimum size=13mm,fill=white] (YD) [below of=WD] {$Y_{\ns-1}$};
  \node[circle,draw,minimum size=13mm,fill=white] (YE) [below of=WE] {$Y_\ns$};
  
  \begin{scope}[on background layer]
  \draw[->] (YA) edge[densely dotted, thick] (WB);
  \draw[->] (YA) edge[densely dotted, thick] (WBB);
  \draw[->] (YA) edge[densely dotted, thick] (WDD);
  \draw[->] (YA) edge[densely dotted, thick] (WE);
  \draw[->] (YB) edge[densely dotted, thick] (WBB);
  \draw[->] (YB) edge[densely dotted, thick] (WDD);
  \draw[->] (YB) edge[densely dotted, thick] (WD);
  \draw[->] (YB) edge[densely dotted, thick] (WE);
  \draw[->] (YBB) edge[densely dotted, thick] (WDD);
  \draw[->] (YBB) edge[densely dotted, thick] (WD);
  \draw[->] (YBB) edge[densely dotted, thick] (WE);
  \draw[->] (YDD) edge[densely dotted, thick] (WD);
  \draw[->] (YDD) edge[densely dotted, thick] (WE);
  \draw[->] (YD) edge[densely dotted, thick] (WE);
  \end{scope}
  
  \node (P) [above of=C] {$\p$};
  \node[rectangle,draw, minimum size=10mm] (R) [below of=YC] {Server};
  \node (out) [below of=R,node distance=13mm] {output};

  \draw[->] (P) edge[densely dashed,bend right=10] (A)(A) edge (WA)(WA) edge (YA)(YA) edge[bend right=10] (R);
  \draw[->] (P) edge[densely dashed,bend right=5] (B)(B) edge (WB)(WB) edge (YB)(YB) edge[bend right=5] (R);
  \draw[->] (P) edge[densely dashed] (BB)(BB) edge (WBB)(WBB) edge (YBB)(YBB) edge (R);
  \draw[->] (P) edge[densely dashed] (DD)(DD) edge (WDD)(WDD) edge (YDD)(YDD) edge (R);
  \draw[->] (P) edge[densely dashed,bend left=5] (D)(D) edge (WD)(WD) edge (YD)(YD) edge[bend left=5] (R);
  \draw[->] (P) edge[densely dashed,bend left=10] (E)(E) edge (WE)(WE) edge (YE)(YE) edge[bend left=10] (R);
  \draw[->] (R) edge (out);
\end{tikzpicture}}
\caption{The information-constrained distributed model. In the
private-coin setting the channels $W_1,\dots,W_\ns$ are independent,
while in the public-coin setting they are jointly randomized, and in
the interactive setting $W_\tst$ can also depend on the previous
messages $Y_1, \ldots, Y_{\tst-1}$ (dotted, \newest{upwards} arrows).} 
\label{fig:model}
\end{figure}
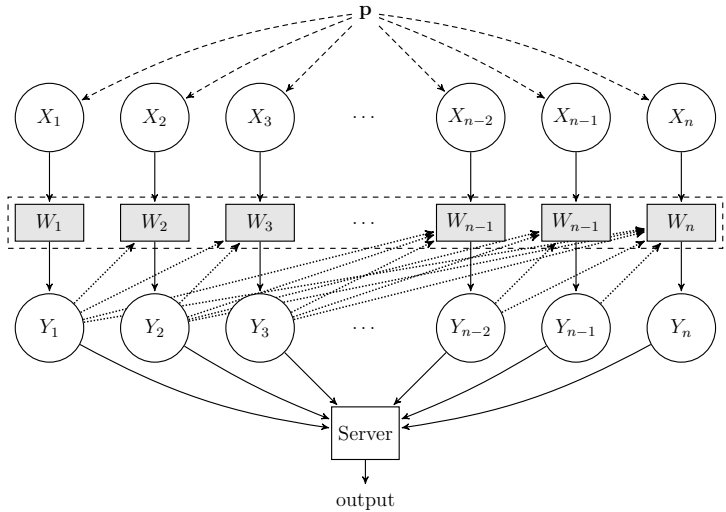

The general setting is captured in~\cref{fig:model}. There are $\ns$
users, each of which observes an independent sample from an unknown
distribution $\p$ over $[2\ab]=\{1,2,\dots,2\ab\}$.\footnote{For
convenience, we assume throughout the paper that the domain $\cX$ has
even cardinality; specifically $\cX=[2\ab]$. This is merely for the ease
of notation, and all results apply to any finite domain
$\cX$.} Each
user is constrained in the amount of information they can
reveal about their input. This constraint for user $\tst$ is described
by a channel $W_\tst\colon[2\ab]\to\cY$, which 
is a randomized
function from $[2\ab]$ to the message space $\cY$.\footnote{Throughout, we use
the information-theoretic notion of a channel
and
use the standard notation $W(y\mid x)$
for the probability with which the output is $y$
when the input is $x$.}
In general, we will
consider a set of channels $\cW$ from which each user's channel must
be selected; this family of ``allowed channels'' models the local
information constraints under consideration.
This is a very general
setting, which captures communication and local privacy constraints as
special cases, as we elaborate next.
\begin{description}
  \item[Communication constraints.] Let $\cW_\numbits\eqdef \{W\colon[2\ab]\to\{0,1\}^\numbits\}$ be the set of channels whose output alphabet $\cY$ is the set of all $\numbits$-bit strings. This captures the constraint where the message from each user is at most $\numbits$ bits: that is, each user has a stringent bandwidth constraint. 
  
  \item[Local differential privacy constraints.] For a privacy parameter $\priv>0$, a channel $W\colon [2\ab]\to\{0,1\}^\ast$ is \emph{$\priv$-locally differentially private}~\cite{EvfimievskiGS:03,DMNS:06,KLNRS:11} if 
\[
\frac{W(y\mid x_1)}{W(y\mid x_2)} \leq
  e^{\priv}, \quad\forall x_1,x_2\in[\newest{2}\ab], \forall
  y\in \{0,1\}^\ast.\,
\]
Loosely speaking, no output message from a user can reveal too much about their sample. We denote by $\cW_\priv$ the set of all $\priv$-locally differentially private ($\priv$-LDP) channels.
\end{description}

We emphasize that, although these two constraints will be
our leading examples, our formulation of local information
constraints captures many more settings. As an example, choosing
message output $\cY = [2\ab]\cup\{\bot\}$ and $\cW$ to be the set
$W\colon[2\ab]\to\cY$ of the form $W(x\mid x) = \eta_x$, $W(\bot\mid
x) = 1-\eta_x$ for various sequences $(\eta_x)_{x\in[2\ab]}$ lets one
model \emph{erasure channels}.
As another example,
\hmargin{Removed the total order comment. It is clear that the elements of $[2\ab]$
are ordered. }
one can choose $\cY=\{0,1\}$, and
let $\cW$ to be the set of channels of the form $W(1\mid x)
= \indic{x\leq \tau}$, \ie{} of \emph{threshold measurements}.
\smallskip

We now return to the description of distributed inference protocols under local information constraints described by $\cW$. Once the channel $W_\tst\in\cW$ at user $\tst$ is decided, the message of user $\tst$ is $y\in\cY$ with 
probability $W_\tst(y\mid X_\tst)$. The transcript of $\ns$ messages,
$Y^\ns=(Y_1,\dots, Y_\ns)$, 
is observed by a server $\mathcal{R}$, whose goal is to perform some
inference task based on the messages.
We consider three classes of protocols,
classified depending on how the channels are allowed to be chosen.\footnote{In what follows, ``SMP'' stands for
  \emph{simultaneous-message passing}, \ie{} for noninteractive,
  one-shot protocol.}

\begin{description}
\item[Private-coin noninteractive (SMP) protocols.]  Let $U_1,
  \dots, U_\ns$ be independent random variables which are independent jointly of $(X_1, \dots,
  X_\ns)$. $U_t$ is available \newest{only} at user $\tst$ and $W_\tst$ is chosen as a
  function of $U_\tst$. Therefore, the outputs of the channels are independent of
  each other.

\item[Public-coin noninteractive (SMP) protocols.] Let $U$ be a random
  variable independent of $(X_1, \dots, X_\ns)$. All users are given access to $U$,
  and they select their respective channels $W_\tst\in \cW$ as a
  function of $U$. We note that the outputs of the channels are
  independent given $U$.  %
 
\item[Sequentially interactive protocols.] Let $U$ be a random
  variable independent of $(X_1, \dots, X_\ns)$. In an
  \emph{interactive} protocol, all users are given access to $U$, and
  user $\tst$ selects their respective channel $W_\tst\in \cW$ as a
  function of $(Y_1, \ldots, Y^{\tst-1}, U)$. We will often make this
  dependence on previous messages explicit by writing $W^{Y^\tst,U}$
  or as $W^{Y^\tst}$ when $U$ is fixed \newest{(see~\cref{sec:preliminaries})}.   
\end{description}

Henceforth, we will interchangeably use ``interactive'' and
``sequentially interactive,'' and will often omit to specify
``noninteractive'' when mentioning public- and private-coin protocols.
\newest{Note that private-coin protocols are a subset of public-coin protocols which in turn are a subset of interactive protocols}.

We now define information-constrained discrete distribution estimation and uniformity testing. For a discrete domain $\cX$, let $\distribs{\cX}$ be the simplex of distributions over $\cX$. Throughout this paper we consider $\cX=[2\ab]$, \newest{and denote $\distribs{[2\ab]}$ by $\Delta_{2\ab}$}.
\begin{description}
  \item[Distribution learning.] In the $(2\ab, \dst)$-distribution learning problem (under constraints $\cW$), we 
    seek to estimate an unknown distribution $\p$ over \new{$\cX=[2\ab]$} to within
    $\dst$ in total variation distance 
    (defined in~\cref{eqn:distances}). Formally, a protocol
    $\Pi\colon[2\ab]^\ns \times \cU \to\cY^\ns$ (using $\cW$) and an estimator mapping
    $\hat{\p}\colon\cY^\ns\times \cU \to \distribs{2\ab}$ constitute an
    $(\ns,\dst)$-estimator using $\cW$ if
\begin{equation}
    \label{eq:def:learning}
\sup_{\p\in\distribs{2\ab}}\probaDistrOf{X^\ns\sim \p}{\totalvardist{\hat{\p}(Y^\ns, U)}{\p}>\dst} \leq \frac{1}{100},
\end{equation}
where \newer{$Y^\ns=\Pi(X^\ns, U)$} and $\totalvardist{\p}{\q}$ denotes the total variation distance between $\p$
and $\q$. Namely, given the transcript $(Y^\ns,U)$ of the protocol $\Pi$ run on the samples $X^\ns$, 
$\hat{\p}$ estimates the input distribution $\p$ to within distance
$\dst$ with probability at least $99/100$ (this choice of
probability is arbitrary and has been chosen for convenience \newer{in the proof of~\cref{lem:paninski:estimation}}).
The \emph{sample complexity} of $(\ab, \dst)$-distribution learning using $\cW$ is then 
the least $\ns$ such that there exists an $(\ns,\dst)$-estimator using $\cW$.

  \item[Identity and uniformity testing.] In the $(2\ab, \dst)$-identity testing problem (under constraints $\cW$), given a known reference distribution $\q$ over $[2\ab]$, and samples from an unknown $\p$, we seek to test if $\p=\q$ or if it is $\dst$-far
from $\q$ in total variation distance. Specifically, an
$(\ns, \dst)$-test using $\cW$ is given by a protocol
$\Pi\colon[2\ab]^\ns\times \cU \to\cY^\ns$ (using $\cW$) and a randomized decision function $T\colon \cY^\ns\times \cU\to \{0,1\}$ such that
\begin{equation}
\probaDistrOf{X^\ns\sim \q^\ns}{T(Y^\ns, U)=0}\geq \frac{99}{100}\,,\qquad
\inf_{\p: \totalvardist{\p}{\q} \geq \dst}  \probaDistrOf{X^\ns\sim \p^\ns} {T(Y^\ns, U)=1} \geq \frac{99}{100},
\end{equation}
where $Y^\ns =\Pi(X^\ns, U)$. 
In other words, after running the protocol $\Pi$ on independent samples $X^\ns$ and public coins $U$,
a decision function  $T$ is applied to the transcript $(Y^\ns, U)$ of the protocol. Overall, the protocol 
should ``accept'' with high constant probability if the samples come
from the reference distribution $\q$ and ``reject'' with high constant
probability if they come from a distribution significantly far from
$\q$. Once again, note that the choice of $1/100$ for probability of
error is for convenience.\footnote{In other words, we seek to solve the composite hypothesis
testing problem with null hypothesis $\mathcal{H}_0 = \{\q\}$ and
composite alternative given by $\mathcal{H}_1
= \setOfSuchThat{ \q'\in\distribs{2\ab} }{ \totalvardist{\q'}{\q}
\geq \dst}$ in a minimax setting, with both
type-I and type-II errors set to $1/100$.} Identity testing for the
uniform reference distribution $\uniform$ over $[2\ab]$ is termed the
$(2\ab, \dst)$-\emph{uniformity testing} problem, and the \emph{sample complexity} of $(2\ab, \dst)$-uniformity testing using $\cW$ is the least
$\ns$ for which there exists an $(\ns,\dst)$-test using $\cW$ for $\uniform$.
\end{description}

\begin{remark}
  We note that our results are phrased in terms of sample complexity, \ie{} the number of users required to perform the corresponding task. Equivalently, this corresponds to minimax lower bounds on rates of convergence (for estimation) or critical radius (for testing).
\end{remark}

  \subsection{Our results}
    \label{ssec:results}

The lower bounds we develop associate to each channel  $W\colon[2\ab]\to\cY$ a
$\ab$-by-$\ab$ positive semidefinite matrix $H(W)$, \newest{which we term the \emph{channel information matrix}}
(see~\cref{eq:channel:matrix}), which captures the ``informativeness''
of the channel $W$. The spectrum of these matrices $H(W)$,
for $W\in\cW$, will play a central role in our results.
In particular,
for a given family of local constraints $\cW$, the
following quantities will be used:
\begin{align*}
\norm{\cW}_{\rm op} &\eqdef \max_{W\in \cW} \norm{H(W)}_{\rm
op}, \tag{maximum operator norm}
\\
\norm{\cW}_{\ast}
&\eqdef \max_{W\in \cW} \norm{H(W)}_{\ast},\tag{maximum nuclear norm}
\\
\norm{\cW}_{F} &\eqdef \max_{W\in \cW} \norm{H(W)}_{F}. \tag{maximum Frobenius norm} 
\end{align*} 
Two key inequalities to interpret our results are  
\begin{equation}\label{eqn:holder}
    \norm{\cW}_{F}^2 \leq \norm{\cW}_{\rm op} \norm{\cW}_{\ast}
  \text{ and }
    \norm{\cW}_{\rm op}\leq \norm{\cW}_{F} \leq \norm{\cW}_{\ast},
\end{equation}
which follow from H\"older's inequality and monotonicity of norms, respectively.

Our results are summarized in~\cref{table:results:lowerbounds}; we
describe and discuss them in more detail below.   

\renewcommand{\arraystretch}{1.75}
\begin{table}[htb!]%
\caption{Lower bounds for local information-constrained learning and
testing. The public- and private-coin bounds were known from previous
work; the interactive bounds all follow from our results. The bound marked by a $(\dagger)$ was previously established in~\cite{BB:20,AJM:20}.
}
\label{table:results:lowerbounds} 
\begin{adjustwidth}{-1cm}{-1cm}
\centering
\begin{tabular}{|c|>{\columncolor{red!10}} c|>{\columncolor{blue!10}} c|>{\columncolor{yellow!10}} c|>{\columncolor{red!10}} c|>{\columncolor{blue!10}} c|>{\columncolor{yellow!10}} c| }
\hline
 & \multicolumn{3}{c|}{Learning}
  & \multicolumn{3}{c|}{Testing} \\\hline & \!\!Private-Coin\!\! &\!\! Public-Coin\! \! &\!\! Interactive\!\!
  & Private-Coin & Public-Coin & Interactive\\\hline\hline 
  General
  & \multicolumn{3}{c|}{\cellcolor{yellow!10}$\frac{\ab}{\dst^2}\cdot \frac{\ab}{\norm{\cW}_{\ast}}$}
  & $\frac{\sqrt{\ab}}{\dst^2}\cdot \frac{\ab}{\norm{\cW}_{\ast}}$
  & $\frac{\sqrt{\ab}}{\dst^2}\cdot \frac{\sqrt{\ab}}{\norm{\cW}_{F}}$
  & $\frac{\sqrt{\ab}}{\dst^2}\cdot \frac{\sqrt{\ab}}{\sqrt{\norm{\cW}_{\ast}\norm{\cW}_{\rm op}}}$
\\\hline
 \!\! Communication\!\!
  & \multicolumn{3}{c|}{\cellcolor{yellow!10}$\frac{\ab}{\dst^2}\cdot \frac{\ab}{2^\numbits}$}
  & $\frac{\sqrt{\ab}}{\dst^2} \cdot \frac{\ab}{2^{\numbits}}$ 
  & $\frac{\sqrt{\ab}}{\dst^2}\cdot \sqrt{\frac{\ab}{2^{\numbits}}}$
  &$\frac{\sqrt{\ab}}{\dst^2}\cdot \sqrt{\frac{\ab}{2^{\numbits}}}$ \\\hline
  Privacy
  & \multicolumn{3}{c|}{\cellcolor{yellow!10}$\frac{\ab}{\dst^2}\cdot \frac{\ab}{\priv^2}$}
  & $\frac{\sqrt{\ab}}{\dst^2} \cdot \frac{\ab}{\priv^2}$ 
  & $\frac{\sqrt{\ab}}{\dst^2}\cdot \frac{\sqrt{\ab}}{\priv^2}$ 
  & \hphantom{\small$(\dagger)$~~~~}$\frac{\sqrt{\ab}}{\dst^2} \cdot \frac{\sqrt{\ab}}{\priv^2}$~~~~\small$(\dagger)$\\\hline
    Leaky-Query
  & \multicolumn{3}{c|}{\cellcolor{yellow!10}$\frac{\ab}{\dst^2}\cdot \sqrt{\ab}$}
  & $\frac{\sqrt{\ab}}{\dst^2} \cdot \sqrt{\ab}$ 
  & $\frac{\sqrt{\ab}}{\dst^2}\cdot \sqrt{\ab}$ 
  & $\frac{\sqrt{\ab}}{\dst^2} \cdot\sqrt[4]{\ab}$\\\hline
\end{tabular}
\end{adjustwidth}
\end{table}

\paragraph{Learning} Our first result concerns distribution learning. We establish a new
technical lemma which relates the mutual information between the
parameters of the distribution to learn and the (adaptively chosen)
messages sent by the users to the nuclear norm
$\norm{\cW}_{\ast}$ of the local constraints (\cref{thm:total}). This
key result, combined with an Assouad-type bound for interactive
protocols, yields the following:

\begin{theorem}
  \label{theo:interactive:lb:learning}
The sample complexity of $(2\ab,\dst)$-distribution learning under local constraints $\cW$ using interactive protocols is 
    \[
        \bigOmega{\frac{\ab^2}{\dst^2\norm{\cW}_{\ast}}}.
    \]
\end{theorem}

This bound matches the known lower bound for learning with noninteractive 
\emph{private-coin} protocols in~\citet{ACT:18:IT1}. 

\italicparagraph{LDP and communication-limited learning.} We now apply this to \newer{local differential} privacy (LDP) and communication constraints. While bounds for these two constraints were presented in prior work, the proofs unfortunately break down for interactive protocols (see~\cref{ssec:previous}).
\label{ssec:learning:ldp:communication}

\begin{corollary}
  \label{coro:learning:lb:ldp}
Let $\priv\in(0,1]$. The sample complexity of interactive  $(2\ab,\dst)$-distribution learning under $\priv$-LDP channels $\cW_\priv$ is 
    \[
     \bigOmega{\frac{\ab^2}{\dst^2\priv^2 }}.
    \]
\end{corollary}
\begin{proof}
  This follows from $ \norm{\cW_\priv}_{\ast} = O(\priv^2)$, which was seen in \cite[Lemma~V.5]{ACT:18:IT1}.
\end{proof}
\begin{corollary}
  \label{coro:learning:lb:comm}
For $1\leq \numbits\leq \log\ab$,  the sample complexity of interactive $(2\ab,\dst)$-distribution learning under communication constraints $\cW_\numbits$  is
\[
        \bigOmega{\frac{\ab^2}{\dst^2 2^{\numbits} }}.
    \]
\end{corollary}
\begin{proof}
  This follows from $\norm{\cW_\numbits}_{\ast} \leq 2^\numbits$, which was seen in \cite[Lemma~V.1]{ACT:18:IT1}.
\end{proof}

Both~\cref{coro:learning:lb:ldp,coro:learning:lb:comm} are optimal up to constant factors. In fact there exist noninteractive private-coin protocols that achieve these bounds (see references in~\cref{ssec:previous}), showing that for learning with communication and LDP constraints, interactive protocols are no more powerful than noninteractive ones. 

\italicparagraph{Learning under $\lp[2]$ distance.} Finally, we note that one can instantiate the distribution learning question in~\cref{eq:def:learning} with other distance measures than total variation, \eg{} the $\lp[2]$ distance defined by $\lp[2](\p_1,\p_2) = \normtwo{\p_1-\p_2}$. Our results on total variation distance readily imply the following corollary for $\lp[2]$, which retrieves the two lower bounds from~\cite{BCO:20, BHO:19} and matches the bounds of~\cite{DJW:13,Bassily:19} for LDP in the noninteractive case.
\begin{corollary}
  \label{cor:l2:learning}
For $1\leq \numbits\leq \log\ab$ and $\priv\in(0,1]$, the sample complexities of interactive $(2\ab,\dst)$-distribution learning in $\lp[2]$ distance under constraints $\cW_\numbits$ and $\cW_{\priv}$ are
\[
        \bigOmega{\frac{\ab}{\dst^2 2^{\numbits}}\land \frac{1}{\dst^4 2^{\numbits}}} \text{ ~~~and~~~~} \bigOmega{\frac{\ab}{\dst^2 \priv^2}\land \frac{1}{\dst^4 \priv^2}},
\]
respectively.   
\end{corollary}
\noindent Details can be found in~\cref{sec:learning}.

\paragraph{Testing} Our next result, proved in~\cref{sec:testing}, is a general lower bound for uniformity testing (and
thus, \textit{a fortiori}, on the more general problem of identity
testing).\footnote{As uniformity testing is a special case of identity
testing, lower
bounds for the former problem imply worst-case lower bounds for the
latter. 
}

\begin{theorem}
  \label{theo:interactive:lb:testing}
The sample complexity of $(2\ab,\dst)$-uniformity testing under local constraints $\cW$ using interactive protocols is 
    \[
        \bigOmega{\frac{\ab}{\dst^2 \sqrt{\norm{\cW}_{\rm op} \norm{\cW}_{\ast}}}}.
    \]
\end{theorem}

\noindent\citet{ACT:18:IT1} previously established an $\bigOmega{\frac{\ab}{\dst^2\norm{\cW}_{F}}}$  lower
bound for (noninteractive) public-coin protocols.

\italicparagraph{LDP and communication-limited testing.}  We now apply this to common local constraints.
\amargin{I commented out the thing with norm comments, since I think they are written below the results. We can of course move them here.}
\begin{corollary}
  \label{coro:testing:lb:ldp}
Let $\priv\in(0,1]$. The sample complexity of interactive  $(2\ab,\dst)$-uniformity testing under $\priv$-LDP channels $\cW_\priv$ is 
    \[
        \bigOmega{\frac{\ab}{\dst^2\priv^2 }}.
    \] 
\end{corollary}
\begin{proof}
    This follows from~\cref{theo:interactive:lb:testing} and
 the fact that  $\norm{\cW_\priv}_{\rm op} \asymp \norm{\cW_\priv}_{F} \asymp \norm{\cW_\priv}_{\ast} = O(\priv^2)$ shown in \cite[Lemma~V.5]{ACT:18:IT1}.
\end{proof}
\begin{corollary}
  \label{coro:testing:lb:comm}
Let $1\leq \numbits\leq \log\ab$. The sample complexity of interactive  $(2\ab,\dst)$-uniformity testing under communication constraints $\cW_\numbits$  is
    \[
        \bigOmega{\frac{\ab}{\dst^2 2^{\numbits/2} }}.
    \]
\end{corollary}
\begin{proof}
  This follows from $\norm{\cW_\numbits}_{\ast} \leq 2^\numbits$ (\cite[Lemma~V.1]{ACT:18:IT1}) and $\norm{\cW_\numbits}_{\rm op} \leq 2$ from~\cref{lemma:general:eigenvalue:bound} in~\cref{ssec:testing:separation}.
\end{proof}

Both~\cref{coro:testing:lb:ldp,coro:testing:lb:comm} are tight up to constant factors, as they are in particular achieved by (noninteractive) public-coin protocols~\cite{ACFT:19,ACT:19}). This shows that for communication and local privacy constraints, interactive protocols are no more powerful than public-coin protocols, which are themselves more powerful than private-coin protocols.

\italicparagraph{A separation.} By relations between matrix norms~\eqref{eqn:holder}, it can be seen that the noninteractive public-coin lower bound of $\bigOmega{\frac{\ab}{\dst^2\norm{\cW}_{F}}}$ from~\cite{ACT:19:COLT} can be up to a $\ab^{1/4}$ factor smaller than the bound in~\cref{theo:interactive:lb:testing} for interactive protocols. Guided by the analysis of the proof of~\cref{theo:interactive:lb:testing}, we show that this maximal separation is achievable, and in particular demonstrate a separation between noninteractive and interactive protocols for uniformity testing (see~\cref{ssec:testing:separation} for details).
\begin{theorem}  
There exists a natural family of constraints, which we term \emph{leaky-query} channels, under which the sample complexity of $(2\ab,\dst)$-uniformity testing for noninteractive public-coin protocols and interactive protocols are $\Theta(\ab/\dst^2)$ and $\Theta(\ab^{3/4}/\dst^2)$, respectively. 
\end{theorem} 

\italicparagraph{Power of the proof.}
Finally, we emphasize that $\sqrt{\norm{\cW}_{\rm
op} \norm{\cW}_{\ast}}$ is a convenient, easy-to-apply bound which is optimal for the channel families considered above. However, the power 
of our techniques goes beyond that specific evaluation. To show this, we provide in~\cref{ssec:testing:loosebound:erasure} a family of partial erasure
constraints $\cW_{\bot}$ for which the bound given
in~\cref{theo:interactive:lb:testing} is loose, and for which
interactivity does not help.
Yet, while the general bound given in the statement of the theorem is not tight, the {proof} of~\cref{theo:interactive:lb:testing}, instantiated with this specific family $\cW_{\bot}$ in mind, readily gives the correct bound.

  \subsection{Prior work}
    \label{ssec:previous}
  There is a vast literature on statistical inference under LDP and 
communication constraints. We discuss some of these works below, focusing on
those most relevant to ours.     

Several protocols have been proposed for discrete distribution estimation and 
testing under communication and privacy constraints. To the best of
our knowledge, all these schemes  
are noninteractive.~\cite{EPK:14, DJW:13:FOCS, WHWNXYLQ:16, KBR:16,
ASZ:18, YeB17, AS:19} provide schemes under LDP, and~\cite{HMOW:18, HOW:18:v1, ACT:18:DSDI,
ACT:19:ICML} provide estimation schemes under
communication constraints.~\cite{Bassily:19} considers estimation schemes under LDP in the $\lp[2]$ distance.~\cite{ACFT:19, BB:20, AJM:20} consider
distribution testing under various privacy constraints,
and~\cite{ACT:19:COLT, ACT:18:IT1, ACT:19, ACT:19:ICML, FMO:18} study distribution testing under
{several} communication constraints.~\cite{ACHST:20} focuses on the role of shared randomness in distributed testing under information constraints. 
Most relevant to this paper is prior work by a subset of the authors~\cite{ACT:18:IT1}
which 
provides a unifying view of lower bounds under information
constraints in the noninteractive setting. We build on this work here.

\paragraph{Interactive testing and estimation of discrete distributions}
We now describe prior work on distribution testing and learning for
discrete distributions in the interactive setting. We focus on the papers that obtain or claim similar
results as ours. We point out the technical flaws in some of the
prior work and outline the state of the art.

\citet{DJW:17} (also, see preprint~\cite{DJW:13}) state 
lower bounds on distributed estimation
of several families of distributions under LDP constraints. 
While their results hold true in one-dimensional settings and for noninteractive protocols, a crucial component of their proof of a private analogue of Assouad's method (\cite[Proposition 3]{DJW:17}) is their claim that, under a marginal mixture distribution they consider, the distribution of the sample is independent from the previous messages; in particular, this claim is used to show \cite[Theorem 3, supplemental (12)]{DJW:17}. 
This claim of independence simplifies
the analysis and
takes care of several dependency
structures that might arise in  sequentially interactive
protocols.
However, this key identity only holds for noninteractive protocols, not in general, even in the absence of any local constraint.

In another direction,~\citet{HMOW-ISIT:18} claim lower bounds on
distributed estimation (in total variation distance)
of several families of distributions under communication constraints.
In~\cite{HOW:18:v1}, the authors claim lower
bounds on distribution estimation under the $\lp[2]$ distance as
well. The arguments
of~\cite{HMOW-ISIT:18} and~\cite{HOW:18:v1} appear to both rely on a
particular flawed step stated as~\cite[Lemma~3]{HOW:18:v1}, which essentially reduces their problem to the
noninteractive setting. But this step does
not hold in the interactive setting. Following an earlier version of this work, made available as preprint, the authors of~\cite{HOW:18:v1} were able to mend their proofs by leveraging the techniques developed in the present paper.

Turning to testing, both \citet{AJM:20} and~\citet{BB:20} establish 
optimal lower bounds on uniformity testing under LDP constraints.\footnote{The conference version of~\cite{AJM:20} had a flaw on the claim that for any fixed setting
of the previous messages $Y^{\tst-1}$, the random variables $Y_\tst$
and $Z_i$ (a parameter of their lower bound construction) are
independent conditioned on $X_\tst\notin\{2i-1, 2i\}$. This is not true in general, as it overlooks some conditional
dependencies that may arise. However, after this was brought to their attention, the authors were able fix the gap in their argument, using techniques that bear some resemblance to the ones used in the present paper~\cite{AJM:20:personal}.}\footnote{The result of~\cite{BB:20} is
actually phrased in a more general way, as they address the more
general problem of \emph{identity testing}, where the reference
distribution need not be uniform and the lower bound quantitatively depends on the reference distribution itself.}
In particular, this implies that the separation between
private-coin and public-coin noninteractive LDP protocols shown
in~\cite{ACT:19:COLT} does not increase when allowing sequential
interactivity.
\hmargin{Why are we giving them credit for general bounds? Rephrased to
take some credit away:)}
However, we note that the results in these paper do not extend
to general constraints. Furthermore, even when we try to use
their techniques to obtain general bounds, we only get
 bounds as a function of
$\norm{\cW}_\ast$ alone. This turns out to be optimal for LDP
constraints, but would lead a suboptimal bound for other types of
constraints, such as communication constraints (where it would yield a denominator
of $2^\numbits$ instead of the optimal $2^{\numbits/2}$).

Several works have studied distribution estimation under the $\lp[2]$
distance~\cite{BGMNW:16, BHO:19, BCO:20} for parametric families of distributions.~\cite{BHO:19,
BCO:20} develop Fisher information-based methods to obtain these
bounds, and one of the distribution families they consider is the
class of discrete distributions. For sequentially interactive
protocols and this particular class, our lower bounds for estimation
under the total variation distance imply their results.  We point out,
nevertheless, that their bounds apply to a larger class of protocols
(\emph{blackboard} protocols, which allow for multiple rounds of
messages from each user), and therefore hold in a more general setting
than ours. However, it is important to remark that these techniques do
not suffice for the total variation distance setting, and more
importantly, unlike the bounds claimed in~\cite{HOW:18:v1}, cannot
give bounds for testing.

Slightly further from the setting considered
here,~\cite{DGLNOS:17}\footnote{To the best of our knowledge, the
details of the proofs of~\cite{DGLNOS:17} have not been made publicly
available, and as such we have not been able to assess correctness of
the results claimed in this paper.} and~\cite{DGKR:19} also consider
distributed estimation and identity testing of discrete distributions,
respectively, under total variation distance. Although their results
apply to blackboard protocols, their setting and results are
incomparable to ours as they consider constraints on the \emph{total}
communication sent by all the users.

In a different direction, recent work of~\citet{BCL:20} considers the
task of testing whether a quantum state is \emph{maximally mixed}, the
quantum analogue of uniformity testing. They focus on the setting
where one is only allowed local measurements (\ie{} without
entanglement) and provide a lower bound showing a separation between
sequentially interactive ``local measurement'' protocols and the more
general fully entangled ones. We note that, while the setting differs
from the one we consider here, some of the considerations are similar,
and there is a direct analogy between their techniques and those
of~\cite{ACT:19:COLT}.

\paragraph{General interactive testing and estimation bounds}
Several papers have studied the role of interactivity for specific
estimation tasks, establishing separation results under either local
privacy or communication constraints.

The study of interactivity in LDP started with~\cite{KLNRS:11} who
designed a learning task that requires exponentially fewer samples
with interactive protocols than with noninteractive ones. Moreover,
this separation can manifest itself in very natural optimization or
learning problems~\cite{STU:17,DF:19,Ullman:18}.~\citet{JMNR:19}
and~\citet{JMR:20} study the relation between sequentially interactive
protocols and fully interactive protocols (where the same user can
send multiple messages), establishing both relations and strong
separations in sample complexity between the two
settings.~\citet{DR:19}, drawing on machinery from the communication
complexity literature, develop a lower bound results which apply to
any locally private estimation protocol (regardless of the
interactivity model).~\citet{Shamir:14} studies various estimation tasks under a range of information constraints. Finally,~\citet{DF:20} establish a separation
between interactive and noninteractive learning for large-margin
classifiers, under both local privacy and communication constraints.

\section{Preliminaries}
  \label{sec:preliminaries}
Hereafter, we write $\log$ and $\ln$ for the binary and natural
logarithms, respectively. We will consider probability distributions
over $[2\ab]$ which we identify with their probability mass functions
$\p\colon[2\ab]\to [0,1]$ satisying $\sum_{x\in\cX} \p(x) = 1$. We denote by $\distribs{2\ab}$ the set of 
all such probability distributions. We denote by $\uniform$ the
uniform distribution over $[2\ab]$.

For two distributions
$\p_1,\p_2$ over $\cX$, denote
their total variation distance by
\begin{equation}
\totalvardist{\p_1}{\p_2} \eqdef \sup_{S\subseteq \cX} (\p_1(S)
- \p_2(S)),
\label{eqn:distances}	 
\end{equation}
and their Kullback--Leibler divergence
and chi square divergence, respectively, by
\begin{equation*}
\kldiv{\p_1}{\p_2} \eqdef \sum_{x\in\cX} \p_1(x) \log\frac{\p_1(x)}{\p_2(x)} \quad\text{ and }\quad
\chisquare{\p_1}{\p_2}\eqdef\sum_{x\in\cX} \frac{(\p_1(x) - \p_2(x))^2}{\p_2(x)}.
\end{equation*}
By Pinsker's inequality and concavity of logarithm, these quantities
obey the inequalities:
\[
\totalvardist{\p_1}{\p_2}^2 \leq \frac{\ln
2}{2} \kldiv{\p_1}{\p_2} \leq \frac{\ln 2}{2} \chisquare{\p_1}{\p_2}\,.
\]

Throughout, We use  the standard asymptotic notation $\bigO{f}$,
$\bigOmega{f}$, $\bigTheta{f}$. In addition, we will often write
$a_n\lesssim b_n$  (resp. $a_n\gtrsim b_n$), to indicate there exists
an absolute constant $C>0$ such that $a_n \leq C\cdot b_n$
(resp. $a_n \geq C\cdot b_n$) for all $n$, and accordingly write
$a_n \asymp b_n$ when both $a_n\lesssim b_n$ and $a_n\gtrsim b_n$. 

\paragraph{Interactive protocols}
We set up some notation for sequentially interactive protocols, defined
in~\cref{sec:setting}. When public-coin $U$ is \newest{fixed} \newer{constant}, we will call the protocol a {\em
deterministic protocol}\amargin{are we sure deterministic is the best phrase we could come up with?}.\footnote{This is a slight abuse of notation, since randomness is used by the channels to generate their \newest{(random)} output.}
Recall that in interactive protocols,
user $\tst$ selects its channel $W\in \cW$ as a function of $(Y^{\tst-1}, U)$.
We denote this channel by $W^{Y^{\tst-1},U}$, or simply by $W^{Y^{\tst-1}}$
for deterministic protocols, and the corresponding output by $Y_\tst$.

We call $(Y^\ns, U)$ the {\em transcript} of the protocol, which is used to
complete the  inference task.
For a fixed protocol $\Pi$, when the input $X^\ns$ has distribution
$\p^\ns$, we denote the distribution of the transcript by $\p_\Pi^{Y^\ns, U}$.
In fact, we often omit the dependence on the protocol from our notation (since it will
be clear from the context) and simply use $\p^{Y^\ns, U}$. For deterministic protocols,
$\p^{Y_\tst|Y^{\tst-1}}$ will be used to denote the conditional distribution of the message
$Y_\tst$ of the $\tst$th user, conditioned on the past messages $Y^{\tst-1}=(Y_1, \ldots, Y_{\tst-1})$.

\paragraph{Lower bound construction} 
Our lower bounds rely on a family of perturbed distributions around
$\uniform$, a common starting point for establishing several
statistical lower bounds.\amargin{commented out reference to [act] here}
The particular construction we
use is from~\citet{Paninski:08} and 
\hmargin{No need to call it Paninski's construction; we can just refer to the equation.}
consists of $2^\ab$ distributions parameterized by
$\cZ=\bool^{\ab}$. Specifically, for $z\in\cZ$ the distribution $\p_z$
over $[2\ab]$ is given by 
\begin{equation}\label{eq:paninski} 
\p_z= \frac{1}{2\ab} \left(1+4\dst z_1, 1-4\dst z_1, \ldots, 1+4\dst
z_t, 1-4\dst z_t, \ldots, 1+4\dst z_{\ab}, 1-4\dst z_{\ab} \right)\,. 
\end{equation}
Each such $\p_z$ is therefore at total variation exactly $2\dst$ from $\uniform$. 

\paragraph{The channel information matrix}
We capture the information revealed by a channel about the distribution of its input
in terms of a matrix $H(W)$, which was defined in~\cite[Definition~I.5]{ACT:18:IT1}.

Specifically, for a channel $W\colon [2\ab]\to\cY$,
the associated {\em information matrix} $H(W)$ is 
the $\ab$-by-$\ab$ positive semi-definite (p.s.d.) matrix $H(W)$ given
by 
\begin{equation}
  \label{eq:channel:matrix}
  H(W)_{i_1,i_2} \eqdef \sum_{y\in\cY} \frac{(W(y\mid 2i_1-1)-W(y\mid
  2i_1))(W(y\mid 2i_2-1)-W(y\mid 2i_2))}{\sum_{x\in[2\ab]} W(y\mid
  x)}\,\quad i_1,i_2\in[\ab]. 
\end{equation}
This matrix captures the ability of the channel output to distinguish
between consecutive even and odd inputs, and is thus particularly
tailored to the Paninski perturbed family defined above. However, the ordering
of the elements is arbitrary and we can associate this matrix with any
partition of the domain into equal parts. 

\paragraph{Organization} The remainder of the paper is organized as
follows. In~\cref{sec:learning,sec:testing} we prove our general
results on learning and testing. In~\cref{ssec:testing:separation} we present a set of channels $\cW$ for
which interactivity helps when testing using $\cW$.

\section{Information-loss bounds} 
In this section, we present two bounds relating the loss for estimating Rademacher
random variables using correlated observations to mutual information. The bounds are
simple and highlighted separately here for easy reference -- in essence, they say that {\em small loss implies
large information}, the first step in any information-theoretic lower bound for statistical inference.

First, we consider estimation of a $\{-1,+1\}^\ab$-valued random vector
under the average Hamming loss function $\hamming{u}{v} \eqdef \sum_{i=1}^\ab\indic{u_i\neq v_i}$. Note that
$\bEE{\hamming{U}{V}}=\sum_{i=1}^\ab\bPr{U_i\neq V_i}$.
\begin{lemma}[Hamming loss] \label{lem:information-Hamming-loss}
Consider random variables $(Z,Y)$ with  $Z\in \{-1, 1\}^\ab$ being a random vector with independent Rademacher
  entries. 
  Let $\hat {Z}$ be a randomized function of $Y$, \ie{} such that the Markov relation $Z \text{---} Y
  \text{---} \hat Z$
  holds. Then,  for each $i\in [\ab]$, with $h(t) \eqdef - t\log t -(1-t)\log (1-t)$ denoting the binary entropy function,
 we get
  \[
  \mutualinfo{Z_i}{Y}\geq 1- h\mleft(\bPr{Z_i\neq \hat {Z}_i}\mright),
  \]
  whereby
  \[
\frac 1 \ab\sum_{i=1}^\ab \mutualinfo{Z_i}{Y}\geq 1- h\mleft(\frac 1 \ab\sum_{i=1}^\ab \bPr{Z_i\neq \hat{Z}_i}\mright).
\]
\end{lemma}
  \begin{proof}
Using the data processing inequality for mutual information,
\[
\mutualinfo{Z_i}{Y}= 1- H(Z_i\mid Y)\geq 1 - H(Z_i\mid \hat {Z}_i)\geq 1- h\mleft(\bPr{Z_i\neq \hat {Z}_i}\mright),
\]
where the second inequality is by Fano's inequality. Upon taking average over $i$, we get
\[
\frac 1 \ab\sum_{i=1}^\ab \mutualinfo{Z_i}{Y}\geq 1-\frac 1 \ab\sum_{i=1}^\ab h(\bPr{Z_i\neq \hat {Z}_i})\geq 1 - h\mleft(\frac 1 \ab\sum_{i=1}^\ab \bPr{Z_i\neq \hat{Z}_i}\mright),
\]
where the second inequality holds by concavity of $h(\cdot)$. 
\end{proof}    
  Next, we consider the mean squared loss function
  $\normtwo{u-v}^2$ for $u,v\in \{-1,+1\}^\ab$.
  Denote by $\operatorname{mmse}(Z\mid Y)$ the minimum mean squared
  error for estimating $Z$ using $Y$ given by
\begin{equation}
\operatorname{mmse}(Z\mid Y) \eqdef \min_{g\colon\cY \to \R^\ab} \bEE{\normtwo{Z- g(Y)}^2},
\end{equation}
where the minimum is taken over all randomized functions $g$ of $Y$. It is well known that the minimum is attained by $g(Y)= \bEEC{Z}{Y}$.
\begin{lemma}[Mean squared loss] \label{lem:information-MSE-loss}
Consider random variables $(Z,Y)$ with  $Z\in \{-1, 1\}^\ab$ being a random vector with independent Rademacher
  entries. Then, for each $i\in [\ab]$, we get
\[
\mutualinfo{Z_i}{Y}\geq
\frac{1}{2\ln 2} \bEE{\bEE{Z_i\mid Y}^2},
\]
whereby
\[
\frac 1 \ab\sum_{i=1}^\ab \mutualinfo{Z_i}{Y}\geq
\frac{1}{2\ln 2} \left(1- \frac 1 \ab \operatorname{mmse}(Z\mid Y)\right).
\]
\end{lemma}
\begin{proof}
By using $\bEE{Z_i\bEE{Z_i\mid Y}}=\bEE{\bEE{Z_i\mid Y}^2}$, we obtain
\[
1- \bEE{\normtwo{Z_i-\bEE{Z_i|Y}}^2}
=\bEE{\bEE{Z_i\mid Y}^2}.
\]
Thus, since $\operatorname{mmse}(Z\mid Y)=\sum_{i=1}^\ab \bEE{\normtwo{Z_i-\bEE{Z_i|Y}}^2}$,
the first inequality in the lemma implies the second.

To see the first inequality, note
\[
\mutualinfo{Z_i}{Y}=1 - H(Z_i \mid Y)= \bEE{D(\bPP{Z_i\mid Y}\| \tt{unif})},
\]
where $\tt{unif}$ denotes the Rademacher distribution. Thus, by Pinsker's inequality,
\[
\mutualinfo{Z_i}{Y}\geq \frac{2}{\ln 2}\bEE{\left(\frac 12 - \bPr{Z_i=1\mid Y}\right)^2}
=\frac{1}{2\ln 2}\bEE{\bEE{Z_i\mid Y}^2},
\]
where in the final step we used the observation that for any $\{-1, +1\}$-valued random
variable $V$, $\bEE{V}= 2\bPr{V=1}-1$. This completes the proof of the lemma.
\end{proof}

\section{Interactive learning under information constraints}\label{sec:learning}

We now prove~\cref{theo:interactive:lb:learning}, a lower bound on the
sample complexity for learning using interactive protocols under
general information constraints given by a channel family $\cW$.

We proceed as in~\cite{ACT:18:IT1} and use the construction
in~\cref{eq:paninski}. For each $z\in \bool^\ab$,  let $\p_{z}\in\Delta_{2\ab}$ denote the $2\ab$-ary distribution given in~\cref{eq:paninski}. We use a uniform prior over these distributions to get our lower bound. Specifically, let $Z$ be distributed uniformly over $\bool^\ab$.
Conditioned on $Z$, let $X_1, \ldots, X_\ns$ denote
independent samples from $\p_Z$. 
We run a
(sequentially) interactive protocol $\Pi$ which generates the
messages (transcript) $Y^\ns$ taking values in $\cY^\ns$.
The distribution of messages  over $\cY^\ns$ is given by 
\begin{equation}
  \label{eqn:paninski-message}
\q^{Y^\ns} \eqdef \frac1{2^{\ab}} \sum_{z \in \cZ}\p_{z}^{Y^\ns}
\end{equation}

In \cite{ACT:18:IT1}, Fano's inequality is used to derive the desired
bound. However, this requires us to derive a bound for 
$\mutualinfo{Z}{Y^\ns}$, the joint information in the message about the
vector $Z$. As noted in~\cite{DJW:17}, this is a formidable task for
interactive communication, since the correlation can be rather
complicated. Instead, we exploit the additive structure of total
variation distance to obtain an Assouad-type bound below, which
relates the loss in total variation function to the {\em average
information} $\frac 1 \ab\sum_{i=1}^\ab \mutualinfo{Z_i}{Y^\ns}$.\footnote{Note that $\sum_{i=1}^\ab \mutualinfo{Z_i}{Y^\ns}\leq \mutualinfo{Z}{Y^\ns}$,
suggesting that this bound is perhaps more stringent than the Fano-type bound in~\cite{ACT:18:IT1}.}

\begin{lemma}[Assouad-type bound]
\label{lem:paninski:estimation}
Consider local constraints $\cW$ and $\dst\in(0,1]$. 
Let $(\Pi, \hat{p})$ be an $(\ns,\dst/12)$-estimator using $\cW$ and $(Y^\ns,U)$ be the corresponding transcript.
Then, we must have
  \begin{align}
  \sum_{i=1}^{\ab} \condmutualinfo{Z_i}{Y^{\ns}}{U}\ge \frac{\ab}{2}.
  \end{align}
\end{lemma}
\begin{proof}
The proof involves relating 
PAC-style guarantees provided by~\cref{eq:def:learning} to an expected Hamming-loss guarantee and then applying the information-loss bound in~\cref{lem:information-Hamming-loss}.
Specifically, let
\[
\hat{Z} \eqdef \underset{z\in \bool^{\ab}}{\arg\!\min} \totalvardist{\p_z}{\hat{\p}(Y^\ns, U)}.
\]
By the triangle inequality,
\[
\totalvardist{\p_{\hat Z}}{\p_Z}\leq \totalvardist{\hat{\p}(Y^\ns, U)}{\p_{\hat Z}}
+\totalvardist{\hat{\p}(Y^\ns, U)}{\p_Z}
\leq 2\totalvardist{\hat{\p}(Y^\ns, U)}{\p_Z}
\]
which yields
\[
\probaOf{\totalvardist{\p_{\hat Z}}{\p_Z}> \frac{\dst}{12}}\leq 1/100,
\]
since $\probaOf{\totalvardist{\hat{\p}(Y^\ns, U)}{\p_Z}> \frac{\dst}{10}}\leq 1/100$ by
our assumption for the estimator $(\Pi, \hat p)$.
Noting that $\totalvardist{\p_{z}}{\p_{z^\prime}}\leq 2\dst$ for every $z, z^\prime \in \bool^\ab$,
we get
\begin{align}
\bEE{\totalvardist{\p_{\hat Z}}{\p_Z}}\leq
\frac {99}{100}\cdot \frac{2\dst}{12}+ \frac 1{100}\cdot 2\dst < \frac{\dst}{5}.
\nonumber
\end{align}
\noindent Next, noting that $\totalvardist{\p_z}{\p_{z^\prime}}=(2\dst/\ab) \sum_{i=1}^\ab \indic{z_i\neq z_i^\prime}$,
the previous inequality yields
\begin{align}
\frac{1}{\ab}\sum_{i=1}^\ab \probaOf{\hat{Z}_i \neq Z_i}<\frac 1 {10}.
\nonumber
\end{align}
The proof is now completed using~\cref{lem:information-Hamming-loss}. 
\end{proof}

Upon combining the previous bound with~\cref{lem:paninski:estimation}, we
obtain the proof of~\cref{theo:interactive:lb:learning}.
Interestingly, the same bound will be useful for the testing problem
as well, and is one of the key components of our lower bound recipes
in this paper. We provide its formal proof first, followed by remarks
on its extension (which will be useful) and heuristics underlying the
formal proof.
\begin{theorem}[Average Information Bound] \label{thm:total}
For $\dst\in (0,1/4]$, let $(Y^\ns,U)$ be the transcript of an interactive protocol
using $\cW$, when the input is generated using $\p_Z$ from~\cref{eq:paninski} with a uniform $Z$.
Then, for every $1\leq t\leq \ns$,
\begin{align*}
\frac
1 \ab \sum_{i=1}^{\ab} \condmutualinfo{Z_i}{Y^t}{U}\le \frac{8t\dst^2}{\ab^2}\cdot \norm{\cW}_{\ast}.
\end{align*}
\end{theorem}

\begin{proof}
Since $\sum_{i=1}^{\ab} \condmutualinfo{Z_i}{Y^t}{U}\le \max_u \sum_{i=1}^{\ab} \condmutualinfo{Z_i}{Y^t}{U=u}$,
it suffices to establish the bound for every fixed realization of $U$;
we will assume that $U$ is a fixed constant and the protocol $\Pi$ is a
deterministic interactive protocol.  Fix $1\leq t\leq \ns$ and
consider $i\in [\ab]$.  For the distribution in~\cref{eq:paninski} and
$i\in[\ab]$, let
\begin{equation}
  \label{eqn:paninski-message:partial}
\p_{+i}^{Y^\ns} \eqdef \frac1{2^{\ab-1}} \sum_{z:z_i=+1}\p_{z}^{Y^\ns} \quad \text{ and }\quad \p_{-i}^{Y^\ns} \eqdef \frac1{2^{\ab-1}} \sum_{z:z_i=-1}\p_{z}^{Y^\ns}
\end{equation}
be distributions over $\ns$-message transcripts restricting $z_i$ to
be $+1$ or $-1$. Recalling the definition of $\q^{Y^t}$,
from~\cref{eqn:paninski-message}, we can rewrite
\begin{align}
\q^{Y^t}=\frac{\p_{+i}^{Y^t}+\p_{-i}^{Y^t}}{2}.	
\end{align}
By the convexity of KL divergence,
\begin{align*}
\mutualinfo{Z_{i}}{Y^t} = \frac{{\kldiv{\p_{+i}^{Y^t}}{\q^{Y^t}}+\kldiv{\p_{-i}^{Y^t}}{\q^{Y^t}}}}2\le \frac14\Paren{\kldiv{\p_{+i}^{Y^t}}{\p_{-i}^{Y^t}}+\kldiv{\p_{-i}^{Y^t}}{\p_{+i}^{Y^t}}}.
\end{align*}
For $z\in\bool^\ab$, write $z^{\oplus i}$ for $z$ with the $i$th
coordinate flipped.
 Using the convexity of KL divergence
and applying Jensen's inequality to the right-side of the previous
bound, we get
\begin{align}
\mutualinfo{Z_{i}}{Y^t} \le \frac1{2}\Paren{\frac1{2^\ab}\sum_{z\in\bool^\ab}\kldiv{{\p_z^{Y^t}}}{{{\p_{z^{\oplus i}}^{Y^t}}}}}.\label{eqn:convex-broken}
\end{align}
Now for any $z, z'$, by the chain rule for KL divergence, we have
\begin{align}
\kldiv{\p_z^{Y^t}}{\p_{z'}^{Y^t}}
& = \sum_{\ts=1}^{t} \bE{\p_z^{Y^{\ts-1}}}{ \kldiv{\p_z^{Y_\ts\mid
Y^{\ts-1}}}{\p_{z'}^{Y_\ts\mid Y^{\ts-1}}} }.
\end{align}
Next, we note that
\begin{align}
\label{eq:main:relation:towards:general}
\probaDistrOf{\p_z}{Y_\ts=y \mid Y^{\ts-1}}= \probaDistrOf{\p_{z^{\oplus i}}}{Y_\ts=y\mid Y^{\ts-1}} + \frac{2\dst z_{i}}{\ab}\Paren{W^{Y^{\ts-1}}(y\mid 2i-1)-W^{Y^{\ts-1}}(y\mid 2i)}. 
\end{align}
Indeed, this relation holds since for all $z$
\begin{align*}
\probaDistrOf{\p_z}{Y_\ts=y \mid Y^{\ts-1}}
&= \sum_{j=1}^{\ab}\Paren{\p_z(2j-1)W^{Y^{\ts-1}}(y\mid
2j-1)+\p_z(2j)W^{Y^{\ts-1}}(y\mid 2j)}\nonumber\\ &= \sum_{j\neq
i}\Paren{\p_z(2j-1)W^{Y^{\ts-1}}(y\mid
2j-1)+\p_z(2j)W^{Y^{\ts-1}}(y\mid 2j)}\nonumber\\
&\qquad+ \Paren{\frac{1+4\dst z_{i}}{2\ab}W^{Y^{\ts-1}}(y\mid
2i-1)+\frac{1-4\dst z_{i}}{2\ab}W^{Y^{\ts-1}}(y\mid 2i)}\\
&=\sum_{j\neq i} \Paren{\p_{z^{\oplus i}}(2i-1)W^{Y^{\ts-1}}(y\mid
2j-1)+\p_{z^{\oplus i}}(2j)W^{Y^{\ts-1}}(y\mid 2j)}\nonumber \\
&\qquad+ \Paren{\frac{1-4\dst z_{i}}{2\ab}W^{Y^{\ts-1}}(y\mid
2i-1)+\frac{1+4\dst z_{i}}{2\ab}W^{Y^{\ts-1}}(y\mid 2i)}\\
&\qquad+ \frac{2\dst z_{i}}{\ab}\Paren{W^{Y^{\ts-1}}(y\mid
2i-1)-W^{Y^{\ts-1}}(y\mid 2i)}\nonumber \\
&= \probaDistrOf{\p_{z^{\oplus i}}}{Y_\ts=y\mid Y^{\ts-1}}
+ \frac{2\dst z_{i}}{\ab}\Paren{W^{Y^{\ts-1}}(y\mid
2i-1)-W^{Y^{\ts-1}}(y\mid 2i)}\,.
\end{align*}
Using~\eqref{eq:main:relation:towards:general}, we bound
$\kldiv{\p_z^{Y_\ts\mid Y^{\ts-1}}}{\p_{z^{\oplus i}}^{Y_\ts\mid Y^{\ts-1}}}$ as
follows.  Since the KL divergence is bounded by the chi square
distance, we have
\begin{align}
\kldiv{\p_z^{Y_\ts\mid Y^{\ts-1}}}{\p_{z^{\oplus i}}^{Y_\ts\mid Y^{\ts-1}}}
&\le \sum_{y\in\cY}\frac{\Paren{\probaDistrOf{\p_z}{Y_\ts=y \mid
Y^{\ts-1}}-\probaDistrOf{\p_{z^{\oplus i}}}{Y_\ts=y \mid
Y^{\ts-1}}}^2}{\probaDistrOf{\p_{z^{\oplus i}}}{Y_\ts=y \mid
Y^{\ts-1}}}\nonumber\\
& \le \frac{16\dst^2}{\ab}\sum_{y\in\cY}\frac{\Paren{W^{Y^{\ts-1}}(y\mid
2i-1)-W^{Y^{\ts-1}}(y\mid 2i)}^2}{\sum_{x\in[2\ab]}
W^{Y^{\ts-1}}(y\mid x)}\nonumber \\ & = \frac{16\dst^2}{\ab}
H(W^{Y^{\ts-1}})_{i,i},
\label{eq:each_bit_information}
\end{align}
where we used the observation
\[
\probaDistrOf{\p_{z^{\oplus i}}}{Y_\ts=y \mid Y^{\ts-1}} \geq
\frac{1-2\dst}{2\ab} \sum_{x\in[2\ab]} W^{Y^{\ts-1}}(y\mid x)
\geq \frac{1}{4\ab} \sum_{x\in[2\ab]} W^{Y^{\ts-1}}(y\mid x)
.
\]
It follows that
\begin{align*}
\sum_{i=1}^{\ab} \mutualinfo{Z_i}{Y^t}
& \le\frac{8\dst^2}{\ab} \sum_{i=1}^\ab \Paren{ \sum_{\ts=1}^{t} \bE{\p_z^{Y^{\ts-1}}}{
H(W^{Y^{\ts-1}})_{i,i}}} \\ &
= \frac{8\dst^2}{\ab} \sum_{\ts=1}^{t}\Paren{ \bE{\p_z^{Y^{\ts-1}}}{ \sum_{i=1}^\ab
H(W^{Y^{\ts-1}})_{i,i}}} \\ &
= \frac{8\dst^2}{\ab} \sum_{\ts=1}^{t}\Paren{ \bE{\p_z^{Y^{\ts-1}}}{ \norm{H(W^{Y^{\ts-1}})}_{\ast}}} \\
& \le\frac{8t\dst^2}{\ab}\cdot \norm{\cW}_{\ast},
\end{align*}
concluding the proof.
\end{proof}

\begin{remark}[Information bound for each coordinate]\label{r:each_coordinate_info}
  While we have stated the previous result as a bound for average
  information, our proof gives a bound for information
  $\mutualinfo{Z_i}{Y^t}$ about each coordinate contained in the
  message $Y^t$.  Specifically, by~\cref{eq:each_bit_information} we
  get that
\[
\mutualinfo{Z_i}{Y^t} \leq \frac{8\dst^2}{\ab}\sum_{\ts=1}^{t}\bEE{H(W^{Y^{\ts-1}})_{i,i}}.
\]
This stronger form is useful;
see~\cref{ssec:testing:loosebound:erasure}.
\end{remark}  
\begin{remark}[Is this bound tight?]
An examination of the proof above suggests that the only seemingly
weak bound is~\cref{eqn:convex-broken}. In this step, which is an
important ingredient of our proof and perhaps allows us to circumvent
the difficulty faced by prior works, we simplify the conditional
distribution of $Z_i$ given the past $Y^t$ by conditioning
additionally on all the other coordinates $Z^{-i}=(Z_1,\dots, Z_{i-1},
Z_{i+1},\dots, Z_\ab)$. Our thesis is that until the time $t$ when the
$i$th bit of $Z$ is determined by $Y^t$, the difficulty in determining
$Z_i$ using $Y^t$ is not reduced much even when we condition on all
the other bits $Z^{-i}$. This is a driving heuristic for the bound
above.
\end{remark}

We conclude this section by showing how our proof of~\cref{thm:total} implies the claimed result on estimation
under the $\lp[2]$ distance,~\cref{cor:l2:learning}.
\begin{proofof}{{\cref{cor:l2:learning}}}
Note that, by the Cauchy--Schwarz inequality, a lower bound on estimation for distributions over domain $\cX$ to total variation distance $\dst$ implies a lower bound to $\lp[2]$ distance $\dst/\sqrt{\abs{\cX}}$, \ie{} with a square root of the domain size factor loss in the distance parameter. We will use this to derive our lower bounds under $\lp[2]$ distance: first, by the above it is easy to see that for $0<\dst\leq \frac{1}{4\sqrt{2\ab}}$,~\cref{theo:interactive:lb:learning} implies a lower bound of $\bigOmega{\frac{\ab}{\dst^2\norm{\cW}_\ast}}$ users for learning under any set of constraints $\cW$.

However, for larger values of $\dst$, we cannot directly use the result, as $\sqrt{2\ab}\dst > 1/4$ and our result does not apply. However, we can choose in that case a subset $\cX'\subseteq[2\ab]$ of the domain of size
$|\cX| = 2\flr{1/(32\dst^2)}$, and embed our (total variation) lower bound in this domain. One can check that this will indeed result in a lower bound of $\bigOmega{\frac{\ab}{\dst^2\norm{\cW}_\ast}}$ users, for a $\lp[2]$ distance parameter $\dst$.

Combining the two cases yields a general lower bound for $\lp[2]$ estimation under $\cW$; instantiating the bound to $\cW_\priv$ and $\cW_\numbits$ yields~\cref{cor:l2:learning}.
\end{proofof}

\section{Interactive testing under information constraints}\label{sec:testing}

\subsection{The general bound: Proof of~\cref{theo:interactive:lb:testing}}
We proceed as in~\cite{ACT:18:IT1} and derive a lower bound for testing under information constraints
using Le Cam's two-point method. Specifically, let $Z$ be distributed uniformly over
$\bool^\ab$. Note that for any $(\ns, \dst)$-test $(\Pi, T)$
for $(2\ab, \dst)$-identity testing with
transcript $(Y^\ns, U)$, we must have
\[
\frac 12 \probaDistrOf{\uniform^\ns}{T(Y^\ns, U)=0}+\frac 12 \expect{\probaDistrOf{\p_Z^\ns}{T(Y^\ns, U)=1}}\geq
\frac {99}{100},
\]
where $\p_z$ is given by~\cref{eq:paninski}.
It follows that we can find a fixed realization of $U$ for which the same bound holds; thus, there exists a deterministic interactive protocol $\Pi^\prime$  for which the same bound holds. In the remainder of the section, we will assume that our protocol $\Pi$ is deterministic and denote by $\q^{Y^\ns}$ and $\uniform^{Y^\ns}$, respectively, the probabilities
distribution of the transcript under input distribution $\bEE{\p_Z^\ns}$ and $\uniform^{\ns}$.

Using standard relations between Bayesian error for binary hypothesis testing with uniform prior and the total variation distance, along with Pinsker's inequality, we get that $\kldiv{\q^{Y^\ns}}{\uniform^{Y^\ns}} \geq c$ for a constant $c>0$. It remains to bound this KL divergence, which we do after the following remark.

\begin{remark}[Comparison with decoupled chi square bounds]
  Before proceeding, we draw contrast with the {\em decoupled chi square divergence} bound technique developed~\cite{ACT:18:IT1}. \newest{Their} first step was to bound Kullback--Leibler divergence with chi square divergence and then handle the latter using the so-called ``Ingster's method.'' While very powerful for SMP protocols, this technique requires
  us to handle the correlation of the vector $Y^\ns$ directly, which is a formidable task for interactive protocols.
  Below, we proceed by first applying the chain rule to the Kullback--Leibler divergence
to break it into contribution for each sample
and then bounding it by the chi square divergence. As will be seen below, this allows us to work with one sample at a time.
Further, switching to chi square divergence relates distances between distributions to a bilinear form
involving $H(W)$s. Thus, we can relate distances between distributions to the spectrum of $H(W)$,
a relation that was exploited to establish a separation between public- and private-coin protocols in~\cite{ACT:18:IT1}.
But now we need to handle the posterior distribution of the message $Y_t$ given the past $Y^{t-1}$, under the mixture distribution.
\end{remark}

Proceeding with the proof, by the chain rule for Kullback--Leibler divergence, we can write
\begin{align}
  \label{eq:testing:chainrule}
	\kldiv{\q^{Y^\ns}}{\uniform^{Y^\ns}} &= \sum_{t=0}^{\ns-1} \bE{\q^{Y^{t}}}{\kldiv{\q^{Y_{t+1}\mid Y^{t}}}{\uniform^{Y_{t+1}\mid Y^{t}}}}
\end{align}
We now present the key technical component of our testing bound in the result below.
\begin{lemma}[Per-round divergence bound]
  \label{lemma:at-time-t}
  For every $0\leq t\leq \ns-1$, we have
    \begin{align}
    \bE{\q^{Y^{t}}}{\kldiv{\q^{Y_{t+1}\mid Y^{t}}}{\uniform^{Y_{t+1}\mid Y^{t}}}}
    \le  \frac{4(\ln 2) \dst^2}{\ab}\norm{\cW}_{\rm op}\cdot\sum_{i=1}^{\ab} \mutualinfo{Z_i}{Y^t}.
    \end{align}
\end{lemma}
\begin{proof}
Fix $t$.  As chi-squared divergence upper bounds KL divergence, we have
\begin{align*}
	\bE{\q^{Y^{t}}}{\kldiv{\q^{Y_{t+1}\mid Y^{t}}}{\uniform^{Y_{t+1}\mid Y^{t}}}}
	&\leq \bE{\q^{Y^{t}}}{\chisquare{\q^{Y_{t+1}\mid Y^{t}}}{\uniform^{Y_{t+1}\mid Y^{t}}}} \\
	&= 2\ab\cdot \bE{\q^{Y^{t}}}{\sum_{y\in \cY} \frac{\left(\sum_x W^{Y^{t}}(y\mid x)(\q_{X_{t+1}\mid Y^{t}}(x) - \frac{1}{2\ab}) \right)^2}{\sum_x W^{Y^{t}}(y\mid x)} }.%
\end{align*}
Upon noting that, for all $i\in[\ab]$,
\[
\q_{X_{t+1}\mid Y^{t}}(2i-1)= \frac{1+2\dst \bEEC{Z_i}{Y^{t}} }{2\ab}, \quad
\q_{X_{t+1}\mid Y^{t}}(2i)= \frac{1-2\dst \bEEC{Z_i}{Y^{t}}}{2\ab}, 
\]
we get
\begin{align}
  \bE{\q^{Y^{t}}}{\kldiv{\q^{Y_{t+1}\mid Y^{t}}}{\uniform^{Y_{t+1}\mid Y^{t}}}}
&\leq  \frac{2\dst^2}{\ab} \bE{\q^{Y^{t}}}{ \sum_{y\in \cY} \frac{\left(\sum_{i=1}^k \bEEC{Z_i}{Y^{t}}(W^{Y^{t}}(y\mid 2i-1) - W^{Y^{t}}(y\mid 2i)) \right)^2}{\sum_x W^{Y^{t}}(y\mid x)} } \notag\\
&= \frac{2\dst^2}{\ab} \bE{\q^{Y^{t}}}{ \bEEC{Z }{ Y^{t}}^T H(W^{Y^{t}}) \bEEC{Z}{Y^{t}} }. \label{e:bound:kl}
\end{align}
We \newest{can now bound}\footnote{
  \newest{In view of~\cref{lem:information-MSE-loss}, the right-side of~\cref{e:operator:norm:bound} is large when the mean squared error in estimating $Z$ from $Y^t$ is small.
    Thus, if the divergence in~\cref{e:bound:kl} is large, we should be able
    to determine $Z$ from $Y^t$. 
  }
  }
\begin{align}\label{e:operator:norm:bound}
      \bEEC{Z}{Y^{t}}^T H(W^{Y^{t}}) \bEEC{Z}{Y^{t}} 
      &\leq \norm{H(W^{Y^{t}})}_{\rm op}\cdot\normtwo{\bEEC{Z}{Y^{t}}}^2,
\end{align}
where $\norm{\cdot}_{\rm op}$ denotes the operator norm (or the maximum eigenvalue) of the p.s.d. matrix $H(W^{Y^{t}})$. \smallskip %

\newest{We now take recourse to the information-loss bound in~\cref{lem:information-MSE-loss}
  to relate $\normtwo{\bEEC{Z}{Y^{t}}}^2$ to average information.   
  By combining~\cref{e:bound:kl,e:operator:norm:bound}
  and using~\cref{lem:information-MSE-loss}, we obtain}
\hmargin{It seems that our earlier constant was off by a factor of $4$. Can someone
please check the constant in Lemma 11 and its application here.}
\[
  \bE{\q^{Y^{t}}}{\kldiv{\p^{Y_{t+1}\mid Y^{t}}}{\uniform^{Y_{t+1}\mid Y^{t}}}} \leq \frac{4 (\ln 2)\dst^2}{\ab} \norm{\cW}_{\rm op} \cdot \sum_{i=1}^\ab \mutualinfo{Z_i}{Y^t}\,,
\]
proving the lemma.
\end{proof}

\begin{remark}[Is the bound above tight?]
  A key heuristic underlying our learning bound is the thesis
  that when the information gathered about each coordinate is small,
  the information revealed in the next iteration cannot be too much.
  The bound in~\cref{e:bound:kl} provides a quantitative counterpart for
  this heuristic. 
  The crux of the previous bound is~\cref{e:operator:norm:bound}, which relates
  the Kullback--Leibler divergence to a per-coordinate information quantity
  $\sum_{i=1}^\ab \bEE{\bEE{Z_i\mid Y^{t}}^2}$. As for learning, this enables us to
  circumvent the difficulty in handling the joint correlation between $Z_i$s, when
  conditioned on $Y^t$. In fact, this step can be weak, as we shall see in a later section below.
  Nonetheless, it allows us to relate the distance between message distribution induced by the mixture
   distribution
   and the uniform distribution to the average information quantity of~\cref{thm:total}.
   This connection between learning and testing bounds is interesting in its own right.
\end{remark}

Upon combining~\cref{lemma:at-time-t} with~\cref{eq:testing:chainrule}, summing over $t$,
and using the average information bound of~\cref{thm:total},
we get
\[
 \kldiv{\q^{Y^\ns}}{\uniform^{Y^\ns}} \leq \frac{16 (\ln 2) \dst^4\ns^2}{\ab^2} \norm{\cW}_{\rm op}\norm{\cW}_{\ast},
\]
which gives the desired bound $\ns = \Omega(\ab/(\sqrt{\norm{\cW}_{\rm op} \norm{\cW}_{\ast}}\dst^2))$ for $\kldiv{\q^{Y^\ns}}{\uniform^{Y^\ns}}$ to be $\Omega(1)$. This proves~\cref{theo:interactive:lb:testing}.

\subsection{A bound for $\norm{H(W)}_{\rm op}$}
Next, we record a general property of the matrix $H(W)$, which
is crucial for handling communication constraints but more generally holds for arbitrary information constraints.\amargin{Citing Gershgorin circle theorem.}
\begin{lemma}[Operator-norm bound]
  \label{lemma:general:eigenvalue:bound}
For any channel $W\colon\cX\to\cY$, we have $\norm{H(W)}_{\rm op} \leq 2$.
\end{lemma}
 \begin{proof}
 By the Gershgorin circle theorem, the eigenvalue of a matrix is at most the largest sum of absolute entries of a row. Now, for any $i\in[\ab]$,
 \begin{align*}
 \norm{H(W)}_{\rm op} 
 &\le \sum_{j=1}^{\ab} \abs{ \sum_{y\in\cY} \frac{(W(y\mid 2i-1)-W(y\mid 2i))(W(y\mid 2j-1)-W(y\mid 2j))}{\sum_{x\in[\ab]} W(y\mid x)}} \\
 &\leq \sum_{y\in\cY} \abs{(W(y\mid 2i-1)-W(y\mid 2i))}\Paren{\frac{\sum_{j=1}^\ab\abs{\Paren{ W(y\mid 2j-1)-W(y\mid 2j)}}}{{\sum_{x\in[\ab]} W(y\mid x)}}} \\
 &\leq \sum_{y\in\cY}|W(y\mid 2i-1)-W(y\mid 2i)|
 \le 2\,,
 \end{align*}
 where in the last step we used the fact that $\sum_{y\in\cY} W(y\mid x) = 1$ for all $x\in[2\ab]$.
 \end{proof}
\subsection{The general bound can be tightened}\label{ssec:testing:loosebound:erasure}
We now present a family of channels for which the general lower bound 
of~\cref{theo:interactive:lb:testing} and the true sample complexity
are a factor $k^{1/4}$ apart. Nonetheless, we can follow the proof of
the lower bound instead of directly applying the statement and
establish the tight lower bounds. In other words, the general proof
methodology we have goes beyond the specific form
in~\cref{theo:interactive:lb:testing}.

Let $\cX=[2\ab]$, $\cY\eqdef\cX\cup\{\bot\}$, and $\eta \in
(0,1)$. The family of \emph{partial erasure} channels
$\cW^{\eta}_{\bot}$ consists of $2\ab$ channels from $\cX$ to $\cY$,
indexed by elements of $\cX$ such that for $x^\ast\in\cX$,   
\[
W_{x^\ast}(y\mid x)=
\begin{cases}
1,& \text{if }y=x=x^\ast,
\\
\eta,& \text{if }y=x \text{ and } x\neq x^\ast,
\\
1-\eta,& \text{if }y=\bot \text{ and } x\neq x^\ast.
\end{cases}
\]
Namely, the channel $W_{x^\ast}$ sends the symbol $x^\ast$ exactly and
erases every other symbol $x\neq x^\ast$ with probability
$1-\eta$. Moreover, the channel matrix $H(W_{x^\ast})$
(see~\cref{eq:channel:matrix}) is diagonal with the $i$th diagonal
entry equal to $1+\eta+ \frac{1-\eta}{2\ab-1}$ for
$x^\ast\notin\{2i-1,2i\}$, and is equal to  $2\eta$ otherwise. For
$\eta=1/\sqrt{\ab}$, we can verify that 
\begin{align}
	2 \le \norm{H(W_x)}_F\le 2\sqrt{2},\quad 2\sqrt{\ab} \le \norm{H(W_x)}_\ast \le
2\sqrt{\ab}+2, \quad 1 \le \norm{H(W_x)}_{\rm op}\le 2.
\label{eqn:norm-bounds-partial-erasure}
\end{align}
Using these quantities to evaluate the lower bounds
in~\cref{table:results:lowerbounds}, we get a lower
bound of $\Omega(k/\dst^2)$ for the sample complexity of testing under SMP public-coin
protocols. We now provide a simple SMP private-coin protocols that
achieves this bound.  

We set all the channels to be $W_1$, the channel that erases all
symbols except symbol $x=1$. This can be converted into an erasure
channel with erasure probability $1-\eta$ by simply converting the
$Y_\tst$s that are equal to $1$ to $\bot$ with probability
$1-\eta$. With this modification, the channel output for the users are independent
and identically distributed,
and $\bPr{Y_\tst=x\mid Y_\tst \neq \bot}=\p(x)$, where $\p$ is the underlying distribution. Therefore,
with $O(\frac{\sqrt{\ab}}{\dst^2}\cdot \frac 1\eta)$ users we can
obtain $O(\frac{\sqrt{\ab}}{\dst^2})$ samples and use a
centralized uniformity test. Upon combining  these bounds we get the
following result.

\begin{proposition}\label{t:noninteractive_erasure} 
The sample complexity of noninteractive $(2\ab,\dst)$-uniformity testing under local constraints
  $\cW^{1/\sqrt{\ab}}_{\bot}$ is  $\Theta(\ab/\dst^2)$ for both public-coin and private-coin protocols.
\end{proposition}

Next, using the norm bounds
in~\cref{eqn:norm-bounds-partial-erasure} to evaluate
the general lower bound of~\cref{theo:interactive:lb:testing}
gives a lower bound of $\Omega(\ab^{3/4}/\dst^2)$ for sample
complexity of uniformity testing under $\cW$, using interactive
protocols.

Below we will see that this bound is not tight,
showing that the general bound of~\cref{theo:interactive:lb:testing}
can be loose for specific families. Nonetheless,
we show that the proof
of~\cref{theo:interactive:lb:testing} can be adapted easily to
establish the optimal sample complexity $\Theta(\ab/\dst^2)$ for
interactive protocols, matching~\cref{t:interactive_erasure}. This
will be achieved by an improved evaluation for $\bEEC{Z}{Y^{t}}^T
H(W^{Y^{t}})\bEEC{Z}{Y^{t}}$
in~\cref{e:operator:norm:bound}.

\begin{proposition}\label{t:interactive_erasure} 
Interactive $(2\ab,\dst)$-uniformity testing under local constraints
  $\cW^{1/\sqrt{\ab}}_{\bot}$ requires at least $\Omega(\ab/\dst^2)$ users.
\end{proposition}  
\begin{proof}
We proceed as in the proof of~\cref{lemma:at-time-t} until~\cref{e:bound:kl} and then replace the
bound~\cref{e:operator:norm:bound} with a more precise
one. Specifically, we note that different choices of $x$ simply allow
us to permute the diagonal entries of the diagonal matrix
$H(W_x)$. Therefore, we get
\begin{align*}
&\lefteqn{\bEEC{Z}{Y^{t}}^T H(W^{Y^{t}})
\bEEC{Z}{Y^{t}}}\\
&\leq \left(1+\eta+\frac{1-\eta}{2\ab-1}\right)\max_{1\leq i\leq \ab}\bEEC{Z_i}{Y^{t}}^2
+ \eta \Paren{ \normtwo{\bEEC{Z}{Y^{t}}}^2-\max_{1\leq i\leq \ab}\bEEC{Z_i}{Y^{t}}^2 }
\\
&\leq 2\norminf{\bEEC{Z}{Y^{t}}}^2
+\frac{1}{\sqrt{\ab}} \, \normtwo{\bEEC{Z}{Y^{t}}}^2.
\end{align*}
Combining with Lemma~\ref{lem:information-MSE-loss}, we get
\[
\bEE{\bEEC{Z}{Y^{t}}^T H(W^{Y^{t}})
\bEEC{Z}{Y^{t}}}
\leq 4 (\ln 2)\left(\max_{1\leq i\leq \ab}\mutualinfo{Z_i}{Y^t} + \frac 1{\sqrt{\ab}}\sum_{i=1}^\ab \mutualinfo{Z_i}{Y^t}\right).
\]
Next, we take recourse to~\cref{r:each_coordinate_info} to get a bound
for information $\mutualinfo{Z_i}{Y^t}$ about each coordinate.
We have
\[
\mutualinfo{Z_i}{Y^t} \leq \frac{8\dst^2}{\ab}\sum_{j=1}^{t-1}\bEE{H(W^{Y^{j}})_{i,i}},
\]
which when combined with the previous bound yields
\begin{align*}
\bEE{\bEEC{Z}{Y^{t}}^T H(W^{Y^{t}})\bEEC{Z}{Y^{t}}}
&\leq \frac{32(\ln 2) \dst^2}{ \ab}
\sum_{j=1}^{t-1}
\left(
\max_{1\leq i\leq \ab} \bEE{H(W^{Y^{j}})_{i,i}}+
\frac 1{\sqrt{\ab}}
\sum_{i=1}^\ab \bEE{H(W^{Y^{j}})_{i,i}}\right)
\\
&\leq \frac{320 (\ln 2) \dst^2 (t-1)}{ \ab}.
\end{align*}
It follows from~\eqref{e:bound:kl} that
\[
 \kldiv{\q^{Y^\ns}}{\uniform^{Y^\ns}} \leq 320 (\ln 2)\cdot \frac{\dst^4\ns^2}{\ab^2}, 
 \]
 which completes the proof.
\end{proof}
We close by noting that the proof above provides yet another example of
an application where our lower bound technique yields a tight bound; we believe there can be
many more. We note that even in this example we related the distance
to the per-coordinate information $I(Z_i \wedge Y^t)$. It will
be interesting to seek examples where our technique yields a tight
bound without using an upper bound for~\cref{e:bound:kl} in terms of
per-coordinate information quantities.

\subsection{A separation between non-interactive and interactive protocols}\label{ssec:testing:separation}
We will now show that there exists a ``natural'' family of local
constraints for which the sample complexity of interactive protocols
is much smaller than that of noninteractive protocols for
$(2\ab, \dst)$-uniformity testing. To the best of our knowledge, this
is the first example of a separation between interactive and
noninteractive protocols for a basic hypothesis testing problem.     

\paragraph{The search for a suitable family of channels}
To describe how we identify the family $\cW$ that yields the desired
separation, we first revisit the proof of~\cref{lemma:at-time-t}. 
As noted before, the only possibly loose step in the argument
is~\eqref{e:operator:norm:bound}. As we saw
in~\cref{ssec:testing:loosebound:erasure},
this bound can be improved by carefully examining the spectrum of
$H(W)$ for different $W$s. Our goal in this section is to construct
an example where the bound in~\eqref{e:operator:norm:bound} is tight,
but $\norm{\cW}_F$ is maximally separated from
$\sqrt{\norm{W}_\ast\norm{W}_{\rm op}}$. Towards this, a key
observation we have is that for~\eqref{e:operator:norm:bound} to be
tight, we should have a channel family that such that for each
$\bEEC{Z}{Y^t}$, we can find a channel $W$ such that the maximum eigenvalue
of $W$ is roughly $\norm{W}_{\rm op}$ and it corresponds to an
eigenvector
that is aligned with $\bEEC{Z}{Y^t}$.

Specifically, we seek a set of channels $\cW$ for which (a)~there is a large gap between $\norm{\cW}_F$ and $\sqrt{\norm{\cW}_{\rm
op}\norm{\cW}_{\ast}}$; (b)~there is a noninteractive protocol with
sample complexity $O(\ab/\dst^2 \norm{\cW}_F)$; and (c)~there is an
interactive protocol with  sample complexity
$O(k/\dst^2 \sqrt{\norm{\cW}_{\rm op}\norm{\cW}_{\ast}})$.
In view of the heuristic observation above, we seek $\cW$
such that we can assign the maximum eigenvalue of $H(W)$ in any
direction of our choice, by appropriately choosing $W\in \cW$.
In the
previous section, we designed a set of channels that satisfy (a) and (b). However, (c)~did not hold since~\eqref{e:operator:norm:bound} is
not tight as we could only assign the maximum eigenvalue of $H(W)$ to one of
the standard basis vectors, and to no other direction.

We meet the objectives above with 
channels that release membership queries for particular sets of our choice.
We will
also have a leakage component to introduce an eigenspace with a small
eigenvalue to ensure a large gap in different norms of interest to us.  
We now formalize this family below.      

For $\eta \in [0,1)$, $\vecu\in[0,1]^{2\ab}$, and $\cY\eqdef
[2\ab] \cup \{\signA, \signB\}$,  the \emph{leaky-query} channel $W_u^\eta$ for an
input $x\in[2\ab]$ outputs $x$ with probability $\eta$; otherwise it
outputs $\signA$ and $\signB$ with probability $\vecu_x$ and
$1-\vecu_x$ respectively. \newer{For our scheme, we will use $\vecu_x$
that has all the entries for coordinates in a set $S$ equal to $1$
and outside it equal to $0$, corresponding to a membership query for $S$.}
Let $\cW_{\in}^\eta
= \{W_u^\eta:\vecu\in[0,1]^{2\ab}\}$ 
\[
W_{\vecu}(y\mid x)=
\begin{cases}
\eta, & \text{if } y = x, \\
( 1- \eta )\vecu_x, & \text{if } y = \signA, \\
(1 - \eta) (1 - \vecu_x), & \text{if }  y = \signB.
\end{cases}
\]
Throughout this section, we will consider $\eta=1/\sqrt{\ab}$. 
We begin by evaluating the required norms for this family.
\begin{lemma}
\begin{align}
      2 \le\ &\norm{\cW^{1/\sqrt{\ab}}_{\in}}_F \le 2\sqrt{2}, \notag\\
      2\sqrt{\ab} \le\ &\norm{\cW^{1/\sqrt{\ab}}_{\in}}_\ast \leq 2\sqrt{\ab}+2,\label{eq:separation:norms}\\
      \ &\norm{\cW^{1/\sqrt{\ab}}_{\in}}_{\rm op} = 2.  \notag
\end{align}	
\end{lemma}
\begin{proof}
By the definition of $W_\vecu$, we obtain for $i_1,i_2\in[\ab]$ that
\begin{align*}
	&H(W_\vecu)_{i_1, i_2} \\
	&= \sum_{y \in [2\ab] \cup \{\signA, \signB\} }\!\!\!\!\!\!\!\!\!\!\! \tfrac{(W_\vecu(y\mid 2i_1-1)-W_\vecu(y\mid 2i_1))(W_\vecu(y\mid 2i_2-1)-W_\vecu(y\mid 2i_2))}{\sum_{x\in[2\ab]} W_\vecu(y\mid x)}.
\end{align*}
Note that for every $u \in [0,1]^{2\ab}$ and $y \in [2 \ab]$, $\sum_{x\in[2\ab]} W_\vecu(y\mid x) = W_\vecu(y \mid y) = \eta$ and
\begin{align*}
	(W_\vecu(y\mid 2i_1-1)&-W_\vecu(y\mid 2i_1))(W_\vecu(y\mid 2i_2-1)-W_\vecu(y\mid 2i_2)) \\
	&= \begin{cases}
		\eta^2, & \text{if } i_1 = i_2 = \lceil y/2\rceil,\\
		0, & \text{otherwise}.
	\end{cases}
\end{align*}
Further, for $y \in \{\signA, \signB\}$, we have
\begin{align*}
(W_\vecu(y\mid 2i_1-1)&-W_\vecu(y\mid 2i_1))(W_\vecu(y\mid 2i_2-1)-W_\vecu(y\mid 2i_2)) \\
&= (1 - \eta)^2(\vecu_{2i_1-1} - \vecu_{2i_1}) (\vecu_{ 2i_2-1} - \vecu_{ 2i_2}),
\end{align*}
and
\begin{align*}
	\sum_{x\in[2\ab]} W_\vecu(\signA \mid x) &= (1 - \eta)\sum_{x \in [2\ab]}  \vecu_x, \\
	\sum_{x\in[2\ab]} W_\vecu(\signB \mid x) &= (1 - \eta)\sum_{x \in [2\ab]}  ( 1- \vecu_x).
\end{align*}
Upon combining these bounds, we get
	\[
		H(W_\vecu) = 2 \eta I_\ab + (1 - \eta) \delta(u) \delta(u)^T,
	\]
	where for all $i\in[\ab]$,
	\[
	\delta(u)_i = (\vecu_{2i-1} - \vecu_{2i})\sqrt{\frac{2k}{(\normone{u} )(2k - \normone{u})}}.
	\]
From this, it can be verified that $H(W_\vecu)$ has eigenvalues $2\eta
	+ (1  - \eta) \normtwo{\delta(\vecu)}^2$ with multiplicity one
	and $2\eta$ with multiplicity $\ab-1$. Also, we have
$\normtwo{\delta(\vecu)}^2 \le 2$, 
and moreover that equality holds when $|\vecu_{2i-1}	- \vecu_{2i}| = 1$
for all $i\in[\ab]$ and $\sum_{i \in [\ab]}	u(2i) = \ab/2$.
        Setting the value of $\eta$ to be $1/\sqrt{\ab}$ establishes
	the claimed bounds.	
\end{proof}

The following two results establish our claim of a separation between noninteractive schemes and interactive scheme for $(2\ab, \dst)$-uniformity testing under local constraints $\cW^{1/\sqrt{\ab}}_{\in}$.
\begin{proposition} \label{clm:sepa_ni}
Noninteractive $(2\ab,\dst)$-uniformity testing under $\cW^{1/\sqrt{\ab}}_{\in}$ has sample complexity $\Theta(\ab/\dst^2)$, even when the unknown distribution $\p$ has bounded norm $\norminf{\p} \leq 10/\ab$.
\end{proposition}
\begin{proof}
	The lower bound can be shown by plugging
	$\lVert{\cW^{1/\sqrt{\ab}}_{\in}}\rVert_F \le 2\sqrt{2}$ in the lower
	bound for noninteractive schemes obtained
	in~\cite{ACT:18:IT1}
	(see~\cref{table:results:lowerbounds}). Crucially, this lower bound is established by considering the family of distributions given in~\cref{eq:paninski}, and thus still applies under the promise that $\norminf{\p} \leq \frac{10}{\ab}$. For the upper bound,
	note that with 
	probability  $1/\sqrt{\ab}$ we observe a sample from $[2\ab]$
	from the underlying distribution. Ignoring the binary
	responses and using the same argument as that in the previous section, we
	get a (matching) upper bound for the number of samples needed by this 
	private-coin SMP protocol. 
\end{proof}
Our next result provides the last piece to establish our separation,
by showing that interactive protocols can do strictly better than the
noninteractive ones.\footnote{For simplicitly, we only provide a
protocol for the case of $\dst = \Omega(1/\ab^{1/8})$, which is enough
for our purposes. We believe that handling smaller values of the
distance parameter is possible, but would require a more
involved protocol and analysis.} 
\begin{proposition} \label{clm:sepa_i}
For $\dst \geq 8/\ab^{1/8}$,  interactive $(2\ab,\dst)$-uniformity testing under 
  $\cW^{1/\sqrt{\ab}}_{\in}$, when the unknown distribution $\p$ satisfies $\norminf{\p} \leq 10/\ab$, has sample complexity $\Theta(\ab^{3/4}/\dst^2)$.
\end{proposition}
\begin{proof}
The lower bound can be obtained by plugging 
the bounds $\lVert\cW^{1/\sqrt{\ab}}_{\in}\rVert_\ast \leq 2\sqrt{\ab}+2$
and $\lVert\cW^{1/\sqrt{\ab}}_{\in}\rVert_{\rm op} = 2$
into~\cref{theo:interactive:lb:testing} (noting that the lower bound instances (\cref{eq:paninski}) satisfy the $\lp[\infty]$ promise).  
For the upper bound, we first present a high-level overview of our scheme, and then provide the details. 

\paragraph{Sketch of the scheme} Observe that when we sample
$X\sim\p$ for a distribution $\p$ over $[2\ab]$, we have $\bEE{\p(X)} = \sum_x \p(x)^2 =\normtwo{\p}^2$.
If
$\totalvardist{\p}{\uniform}\ge\dst$, then by the Cauchy--Schwarz
inequality $\normtwo{\p}^2 \ge (1+4\dst^2)/2\ab$, whereas
$\normtwo{\uniform}^2=1/\ab$. Therefore, when we sample from a
distribution that is far from uniform, the expected probability of the
observed sample is also larger than that under the uniform
distribution. Our protocol exploits this and proceeds in two
stages. In the first stage, we select any channel in
$\cW_\in^{1/\sqrt{\ab}}$ and use it for a fraction of the users.
Let $S$ be the set of outputs from these channels that are in
$[2\ab]$. Now, $\uniform(S)=|S|/2\ab$, but using the motivation above we can hope that, for a $\p$ that is far from
$\uniform$, $\p(S)$ will be noticeably
larger. In the next stage, the remaining players choose
$\vecu_x=\indic{x\in S}$. Now, $\uniform(\signA)=(1-\eta)\uniform(S)$,
and $\p(\signA)=(1-\eta)\p(S)$, and we will perform a binary
hypothesis test to separate these two cases.

\paragraph{Detailed argument} The rest of the argument makes the
intuition above formal. We assume that there are $\ns=C\ab^{3/4}/\dst^2$ users, for some
constant $C>0$ which can be taken to be %
$C = 625$ and $\ab \geq \clg{{9C^2}/{2^{16}}} = 54$. 
Our protocol proceeds as follows:
\begin{enumerate}
	\item For the first $\ns/2$ users, choose the channel
	$W_{\mathbf{0}}$, corresponding to a simple erasure channel with erasure probability $1-\eta$. Gather a set $S\subseteq[2\ab]$ of ``leaked'' samples in
	this stage. 
	\item For the last $\ns/2$ users, choose the channel
	$W_\vecu$ for $\vecu$ corresponding to the indicator vector of
	$S$, in order to estimate $\p(S)$ to an additive accuracy \new{$\frac{\dst^2}{4}\bE{\uniform}{\uniform(S)}$}, via the binary responses. 
\end{enumerate}

\italicparagraph{Step 1.} Let $N$ be the number of ``leaked'' samples in the first stage, namely $N$ symbols
are received without erasure in the first stage.  It can be easily checked that
$\bEE{N}= \ns/(2\sqrt{\ab})$ and $\var(N)<\ns/(2\sqrt{\ab})$. Since our assumptions on $C$ and $\ab$ imply that $\ns \geq 800\sqrt{\ab}$, 
by Chebyshev's inequality we have $\ns/(4\sqrt{\ab}) \le N\le 3\ns/(4\sqrt{\ab})$ with probability at least $99/100$; below we proceed assuming this event holds and will account for the probability of its failure at the end.

Let $S\subseteq[2\ab]$ be the set of symbols of ``leaked'' samples in this step; note that $|S|\le N$ (since
some values may be repeated).
By linearity of expectation, when the samples are generated from $\p$, we have 
\begin{equation*}
	\bEE{\p(S)} = \sum_{i=1}^{2\ab} \p(i) \left( 1 - (1-\p(i))^N \right),
\end{equation*}
and $\uniform(S)=1 - (1-1/(2\ab))^N$. When $\p$ is $\dst$-far from $\uniform$, we have $\normtwo{\p}^2 \geq \frac{1+4\dst^2}{2\ab}$, and 
\begin{align*}
	\bE{\p}{\p(S)}  - \bE{\uniform}{\uniform(S)}  
	= \Paren{1-\frac{1}{2\ab}}^N \sum_{i=1}^{2\ab} \p(i) \Paren{1-\Paren{\frac{2\ab}{2\ab-1}(1-\p(i))}^N} 
	\geq \frac{N\dst^2}{\ab},
\end{align*}
where the last inequality follows by using almost the same analysis as that in the proof of~\cite[Lemma~1]{Paninski:08}. Further, we have\cmargin{Follows from $(1-x)^t \leq 1- \frac{tx}{2}$ for $x\in(0,1)$ and $0< t \leq \frac{1}{2x}$.} 
	\[
		\bE{\uniform}{\uniform(S)} =1 - \Paren{1 - \frac{1}{2\ab}}^N \geq \frac{N}{4\ab}.
	\] since $N \leq \ab$ (which follows from our bound $N \leq 3\ns/(4\sqrt{\ab})$, along with $\dst\geq 8/\ab^{1/8}$ and $\ab \geq 9C^2/2^{16}$), and therefore,\cmargin{Need $N\leq \ab$, \ie{} $\ns \leq (4/3)\ab^{3/2}$. This is satisfied for $\ab \geq 9C^2/2^{16}$, as $\dst\geq 8/\ab^{1/8}$.}
\begin{equation*}
	\bE{\p}{\p(S)} \geq (1 + 3\dst^2/2)\bE{\uniform}{\uniform(S)}
\end{equation*}
whenever $\p$ is $\dst$-far from uniform.

Turning to the variance, we can prove the following bound:
\begin{claim}
  \label{claim:variance:bound}
For any $\p$, we have
\[
    \var_\p[\p(S)] \leq \norminf{\p}\bE{\p}{\p(S)}\,.
\]
\end{claim}
\begin{proof}
Denote by $N_1,\dots,N_\ab$ the sample counts, \ie $N_i$ is the number 
of times element $i$ is seen among the $N$ samples. While $N_i$ is 
distributed as a Binomial with parameters $N$ and $\p(i)$, the 
$N_1,\dots,N_\ab$ are not independent; however, they are negatively 
associated (see, \eg~\cite[Section~2.2]{DubhashiR98}), which we will use 
below.
We start with bounding the expected square:
\begin{align*}
  \bEE{\p(S)^2}  
  &= \sum_{i=1}^\ab\sum_{j=1}^\ab \p(i)\p(j) \bEE{\indic{N_i \geq 1}\indic{N_j \geq 1}} \\
  &= \sum_{i=1}^\ab\p(i)^2 \bEE{\indic{N_i \geq 1}} + 2\sum_{i< j} \p(i)\p(j) \bEE{\indic{N_i \geq 1}\indic{N_j \geq 1}} \\
  &\leq \sum_{i=1}^\ab\p(i)^2 \bEE{\indic{N_i \geq 1}} + 2\sum_{i< j} \p(i)\p(j) \bEE{\indic{N_i \geq 1}}\bEE{\indic{N_j \geq 1}} \\
  &= \sum_{i=1}^\ab\p(i)^2 \bPr{N_i \geq 1} + \Paren{\sum_{i=1}^\ab \p(i)\bPr{N_i \geq 1} }^{\!\!2}- \sum_{i=1}^\ab\p(i)^2 \bPr{N_i \geq 1}^2 \\
  &= \sum_{i=1}^\ab\p(i)^2 \bPr{N_i \geq 1}\bPr{N_i = 0} + \bEE{\p(S)}^2\,,
\end{align*}
where the inequality follows from negative associativity, and we got the third equality by completing the sum $2\sum_{i<j} x_{i,j} = \sum_{i,j} x_{i,j} - \sum_{i} x_{i,i}$. Rewriting, we have
\begin{align*}
  \var[\p(S)]
  &\leq \sum_{i=1}^\ab\p(i)^2 \bPr{N_i \geq 1}\bPr{N_i = 0}\,.
\end{align*}
By upper bounding the last factor by $1$, we then get
\begin{align*}
  \var[\p(S)] &\leq \norminf{\p}\expect{\p(S)}\,,
\end{align*}
concluding the proof.
\end{proof}

When $\p=\uniform$, this gives
$\var_\uniform[\uniform(S)] \leq \frac{1}{2\ab}\bE{\uniform}{\uniform(S)}$. By
Chebyshev's inequality, using the chain of inequalities
\[
\frac{2}{\ab\dst^4\bE{\uniform}{\uniform(S)}} 
	\leq \frac{8}{\dst^4 N} 
	\leq \frac{32\sqrt{\ab}}{\dst^4 \ns} 
	\leq \frac{\ab^{3/4}}{2\dst^2 \ns} 
\]
(where the second is due to our lower bound on $N$, and the third follows from our assumption $\dst \geq 8/\ab^{1/8}$)  
we get that
\begin{align}
	\probaDistrOf{X^\ns\sim \uniform^\ns}{
		\uniform(S) <
		(1+\frac{\dst^2}{2})\bE{\uniform}{\uniform(S)}}
	&\geq 9/10 \label{eq:uniform_stat}
\end{align}
as long as $C \geq 5$.

Now, consider the case where $\p$ is $\dst$-far from
uniform. By Chebyshev's inequality, using~\cref{claim:variance:bound} along with the promise that $\norminf{\p} \leq 10/\ab$,  we get
\begin{align*}
	\probaDistrOf{X^\ns\sim \p^\ns}{\p(S) < (1 + \dst^2)\bE{\uniform}{\uniform(S)}} 
	&\leq \probaDistrOf{X^\ns\sim \p^\ns}{\p(S) < \frac{1 + \dst^2}{1+\frac{3}{2}\dst^2}\bE{\p}{\p(S)}}  \\
	&\leq \frac{(2+3\dst^2)^2}{\dst^4} \frac{10}{\ab\bE{\p}{\p(S)}} \leq \frac{1000}{\dst^4 N} \\
	&\leq \frac{125\ab^{3/4}}{2\dst^2\ns}\,,
\end{align*}
where the last inequality is derived as in the uniform case. 
This implies that, if $\totalvardist{\p}{\uniform}\geq \dst$,
\begin{align}
	\probaDistrOf{X^\ns\sim \p^\ns}{\p(S) > (1 + \dst^2)\bE{\uniform}{\uniform(S)}}
	\geq 9/10,
	\label{eq:nonuniform_stat}
\end{align}
as long as $C \geq 625$.

\italicparagraph{Step 2.} In the second stage, the $\ns/2$ users all choose
the channel $W_u$, where $u\in\{0,1\}^{2\ab}$ is the indicator vector
of $S$. We assume now that conditions~\cref{eq:uniform_stat}
and~\cref{eq:nonuniform_stat}, respectively,
hold under the uniform and nonuniform distribution hypothesis.
The goal of this stage is to distinguish between the two cases $\p(S) <
(1+\frac{1}{2}\dst^2)\bE{\uniform}{\uniform(S)}$ and $\p(S) > (1
+ \dst^2)\bE{\uniform}{\uniform(S)}$, which can be done by estimating $\p(S)$
to an additive $\frac{\dst^2}{4}\bE{\uniform}{\uniform(S)}$ with
probability at least $99/100$ (from $\ns/2$ users). Note that this is 
equivalent to estimating the mean of a Bernoulli random variable $p\geq 
\bE{\uniform}{\uniform(S)}$ to an additive $\dst^2 p/8$ from $\ns/2$ 
observations. Using Chebyshev's inequality, we can check $\ns \geq 
320\ab^{3/4}/\dst^2$ suffices.

\italicparagraph{Overall.}  Accounting for the 3
good events above that hold with probability $99/100$, $9/10$, and $99/100$ respectively,
this protocol is correct by a union
bound with probability at least $22/25$, and involves
$\ns=O(\ab^{3/4}/\dst^2)$ users, as desired. By explicit computation of the Binomial distribution probabilities, repeating the protocol 7 times on independent subsets of samples (\ie{} groups of users) and taking the majority output can then boost the success probability from 22/25 to at least 99/100, only changing the total number of samples by  this constant factor 7.
\end{proof}

\bibliographystyle{IEEEtranS}
\bibliography{bibliography} 

\end{document}